\documentclass[11pt]{article}
\usepackage{sectsty}

\usepackage{graphicx}

\usepackage{amsmath,amssymb}
\usepackage{bm}
\usepackage{enumerate}
\usepackage{amsthm}
\usepackage[all,2cell]{xy}
\usepackage{mathrsfs}
\usepackage{mathtools}
\mathtoolsset{showonlyrefs}
\usepackage{tikz-cd}
\usepackage{tikz}
\usetikzlibrary{shapes.multipart}
\usetikzlibrary{patterns,patterns.meta}
\usepackage{xcolor}
\definecolor{green}{rgb}{0.0, 0.5, 0.0}
\DeclareUnicodeCharacter{0301}{*************************************}
\usepackage[normalem]{ulem}

\numberwithin{equation}{section}

\usepackage[colorlinks=true]{hyperref}
\hypersetup{urlcolor=blue, citecolor=red, linkcolor=blue}

\usetikzlibrary{cd}

\usepackage[a4paper, left=25mm, right=25mm, top=30mm, bottom=30mm]{geometry}

\usepackage{marginnote}

\usepackage{imakeidx}
\makeindex[intoc]

\usepackage[nottoc]{tocbibind}
\usepackage[colorinlistoftodos]{todonotes} 

\newtheorem{thm}{Theorem}[section]
\newtheorem{dfn}[thm]{Definition}
\newtheorem{prop}[thm]{Proposition}
\newtheorem{cor}[thm]{Corollary}

\newtheorem{exmpl}[thm]{Example}
\newtheorem{rmrk}[thm]{Remark}

\newcommand\restr[2]{{
  \left.\kern-\nulldelimiterspace 
  #1 
  \right|_{#2} 
}}

\newcommand*{\transp}[2][-3mu]{\ensuremath{\mskip1mu\prescript{\smash{\mathrm t\mkern#1}}{}{\mathstrut#2}}}%

\newcommand{\R}{\mathbb{R}}
\renewcommand{\d}{\mathrm{d}}
\let \dd \d

\newcommand{\Cinfty}{\mathscr{C}^\infty}
\newcommand{\T}{\mathrm{T}}
\newcommand{\cT}{\mathrm{T}^\ast}

\newcommand{\Lie}{\mathscr{L}}

\newcommand{\X}{\mathfrak{X}}

\renewcommand{\L}{L}
\newcommand{\F}{\mathcal{F}}

\newcommand{\Reeb}{R}

\newcommand{\Rt}{\Reeb_{t}}
\newcommand{\Rz}{\Reeb_{z}}
\let\Rs\Rz

\newcommand{\parder}[2]{\frac{\partial #1}{\partial #2}}

\newcommand{\tparder}[2]{\partial #1/\partial #2}
\newcommand{\parderr}[3]{\frac{\partial^2 #1}{\partial #2\partial #3}}

\DeclareMathOperator{\Diff}{Diff}

\DeclareMathAlphabet{\mathpzc}{OT1}{pzc}{m}{it}
\def\d{\mathrm{d}}
\DeclareMathOperator{\id}{id}
\DeclareMathOperator{\pro}{pr}
\newcommand{\pr}[1]{\pro_{#1}}

\newcommand{\liedv}[1]{\Lie_{#1}}
\newcommand*{\contr}[1]{\iota_{#1}}
\newcommand*{\contrp}[1]{\iota\left(#1\right)}

\let\oldemph\emph
\let\emph\textbf

\usepackage[style=ext-numeric,citestyle=numeric-comp,sorting=nyt,sortcites, backend=biber, giveninits=true, maxnames=99]{biblatex}
\bibliography{biblio}
\DeclareFieldFormat[article,periodical]{volume}{\mkbibbold{#1}}
\DeclareFieldFormat[article,periodical]{number}{\mkbibparens{#1}}
\renewbibmacro{in:}{%
  \ifentrytype{article}
    {}
    {\printtext{\bibstring{in}\intitlepunct}}}

\DeclareFieldFormat{issuedate}{#1}
\renewcommand{\jourvoldelim}{\addcomma\space}

\renewbibmacro*{journal+issuetitle}{%
  \usebibmacro{journal}%
  \setunit*{\jourvoldelim}%
  \iffieldundef{series}
    {}
    {\setunit*{\jourserdelim}%
     \printfield{series}%
     \setunit{\servoldelim}}%
  \usebibmacro{volume+number+eid}%
  \setunit{\addcolon\space}%
  \usebibmacro{issue}%
  \setunit{\bibpagespunct}%
  \printfield{pages}%
  \setunit{\space}%
  \usebibmacro{issue+date}%
  \newunit}

\renewbibmacro*{note+pages}{%
  \printfield{note}%
  \newunit}

\DeclareFieldFormat[article,periodical]{date}{\mkbibparens{#1}}

\AtEveryBibitem{
  \ifentrytype{online}{}{
    \ifentrytype{thesis}{}{
        \clearfield{url}
        \clearfield{urldate}
      }
  }
  \ifentrytype{unpublished}{}{ 
  \clearfield{note}
  }
  \clearfield{language}
}

\DeclareSourcemap{
  \maps[datatype=bibtex]{
    \map[overwrite=true]{
      \step[fieldset=urldate, null]
      \step[fieldset=language, null]
      \step[fieldset=address, null]
     \step[fieldset=pagetotal, null]
    }
  }
}

\title{{\sffamily Symmetries, conservation and dissipation \\in time-dependent contact systems}}

\author{{\sffamily 
$^a$Jordi Gaset%
\thanks{e-mail:
   jordi.gaset@upm.es \ ORCID: 0000-0001-8796-3149}\ ,\
$^b$Asier López-Gordón%
\thanks{e-mail:
   asier.lopez@icmat.es \ ORCID: 0000-0002-9620-9647}\ ,\
$^c$Xavier Rivas%
\thanks{e-mail:
   xavier.rivas@unir.net \ ORCID: 0000-0002-4175-5157}\ ,\
}
\\[1ex]
\normalsize\itshape\sffamily
$^a$Departamento de Matemática Aplicada,\\
\normalsize\itshape\sffamily
Escuela Técnica Superior de Ingeniería de Montes, Forestal y del Medio Natural\\
\normalsize\itshape\sffamily
Universidad Politécnica de Madrid, Madrid, Spain.
\\[1ex]
\normalsize\itshape\sffamily
$^b$Instituto de Ciencias Matem\'aticas (ICMAT),\\
\normalsize\itshape\sffamily
Consejo Superior de Investigaciones Cient\'ificas, Madrid, Spain.
\\[1ex]
\normalsize\itshape\sffamily
$^c$Escuela Superior de Ingenier\'{\i}a y Tecnolog\'{\i}a,\\
\normalsize\itshape\sffamily
Universidad Internacional de La Rioja, Logro\~no, Spain.
}
\date{{\sffamily \today}}

\begin{document}

\maketitle

\begin{abstract}
    In contact Hamiltonian systems, the so-called dissipated quantities are akin to conserved quantities in classical Hamiltonian systems. In this paper, we prove a Noether's theorem for non-autonomous contact Hamiltonian systems, characterizing a class of symmetries which are in bijection with dissipated quantities. We also study other classes of symmetries which preserve (up to a conformal factor) additional structures, such as the contact form or the Hamiltonian function. Furthermore, making use of the geometric structures of the extended tangent bundle, we introduce additional classes of symmetries for time-dependent contact Lagrangian systems. Our results are illustrated with several examples. In particular, we present the two-body problem with time-dependent friction, which could be interesting in celestial mechanics. 
\end{abstract}


\noindent\textbf{MSC\,2020 codes:}
70H33; 
37J55, 
53D10, 
53Z05 

{\setcounter{tocdepth}{2}
\def\baselinestretch{1}
\small
\def\addvspace#1{\vskip 1pt}
\parskip 0pt plus 0.1mm
\tableofcontents
}


\section{Introduction}



As it is well-known, symplectic geometry is the natural framework for classical mechanical systems. In the last decades, alternative geometric structures and their associated dynamics have been widely studied. In particular, contact geometry has arisen as a geometric solution to model non-conservative systems \cite{Bravetti2017,Bravetti2017a,deLeon2021b,Lainz,deLeon2019,Gaset2020a,Rivas2022,Gaset2021}, as well as some thermodynamical systems \cite{Simoes2020,Bravetti2019,Eberard2007,Gay-Balmaz2018,Mrugala2000}, quantum systems \cite{Ciaglia2018}, nonholonomic systems \cite{deLeon2021a}, electromagnetism \cite{gaset_application_2022}, gravitation \cite{gaset_variational_2022}, Lie systems \cite{LR-2022}, control theory \cite{deLeon2020a}, dissipative field theories \cite{Gaset2020,Gaset2021a,Rivas2022}, etc.

When a classical mechanical system exhibits explicit time dependence, i.e., it is non-autonomous, its underlying geometric structure can be taken either as a contact structure or as a cosymplectic structure \cite{deLeon2017}. Recently, the so-called cocontact geometry \cite{deLeon2022,RiTo-2022}, a suitable geometric structure describing non-autonomous dissipative systems, combining contact and cosymplectic geometry, has been introduced.

The study of symmetries of mechanical systems is of great interest since it provides a way of finding conserved (or dissipated) quantities. Moreover, reduction procedures can be used in order to simplify the description of a dynamical system whose group of symmetries is known. The relation between symmetries and conserved quantities has been a topic of great interest in mathematical physics since the seminal work by Emmy Noether \cite{Noether1971} (see also \cite{Neeman1999,Kosmann-Schwarzbach2011}).
Since the dawn of geometric mechanics, numerous papers have been devoted to the geometric study of symmetries and conserved quantities for Hamiltonian and Lagrangian systems \cite{Carinena1989,Carinena1989a,Carinena1994,deLeon1989,deLeon1994,deLeon1996,Lunev1990,Djukic1975,Ferrario1990,Marmo1986,Marwat2007,Prince1983,Prince1985,Sarlet1981,Sarlet1983,vanderSchaft1983,deLeon2021c}.
However, in the case of contact (or cocontact) systems, it is more natural to consider the so-called dissipated quantities and their associated symmetries 
\cite{Lainz,Rivas2022}.
Some notions of symmetries for autonomous contact Hamiltonian and Lagrangian systems were independently introduced in \cite{Gaset2020a} and \cite{deLeon2020b}. The study of symmetries and conserved (or dissipated) quantities is also related with Hamilton--Jacobi theory. A first Hamilton--Jacobi equation for autonomous contact systems was obtained in \cite{deLeon2017}, and an alternative one was obtained in \cite{deLeon2021e}. The Hamilton--Jacobi theory for non-autonomous contact systems has been recently done in \cite{deLeon2022a}. Canonical and canonoid transformations \cite{Az1} and Lie integrability \cite{Az2} of (co)contact systems have also been studied.

As a matter of fact, when a (co)contact Lagrangian system exhibits a cyclic coordinate, the associated quantity is no longer conserved but dissipated. In \cite{Bravetti2021} the symmetries and dissipated quantities of time-dependent contact systems were studied.
Their results are restricted to the so-called extended contact phase space, i.e., the extended cotangent bundle $\cT Q\times \R\times \R$ endowed with a contact form defined by the canonical contact form of $\cT Q\times \R$ and the Hamiltonian function of the system. Among the advantages of the cocontact formalism it is the fact that one can consider more general manifolds. Moreover, $\R \times \cT Q\times \R$ is endowed with a canonical cocontact structure, independent of the Hamiltonian function.


In the present paper, the symmetries of time-dependent contact Hamiltonian and Lagrangian systems are studied and classified. A characterization of dissipated quantities and their relation with symmetries is also provided. Firstly, the most general type of symmetries with associated dissipated quantities, the so-called generalized infinitesimal dynamical symmetries, are studied. Secondly, other types of transformations which preserve additional geometric or dynamical structures are discussed, exploring the relations between them.
 After that, we consider symmetries of time-dependent contact Lagrangian systems which also preserve the geometric structures of the extended tangent bundle. Finally, we study three examples in detail: the free particle with time-dependent mass and linear dissipation, the action-dependent central potential with time-dependent mass, and the two-body problem with time-dependent friction. {The latter may have interesting applications in celestial mechanics, allowing to describe the motion of planets with damping provoked by the medium.}

In particular, all our results can be applied to time-independent contact Hamiltonian and Lagrangian systems. We review and extend the results from the literature regarding symmetries in autonomous contact systems \cite{Gaset2020a,deLeon2020b,Gaset2021}. Hence, this paper may also be used as a reference for the reader interested in the symmetries of contact Hamiltonian and Lagrangian systems (even if they do not have an explicit time-dependence).


\bigskip
\textbf{New results and relation to literature.} This paper is, to the best of our knowledge, the first reference studying the symmetries of cocontact Hamiltonian and Lagrangian systems. Cocontact geometry was introduced in \cite{deLeon2022} in order to provide a geometric framework for action and time dependent systems, combining features of contact and cosymplectic geometry.
Furthermore, the present paper may also be used as a reference for the classification of symmetries of autonomous contact Hamiltonian and Lagrangian systems, the relations between them and their associated conserved and dissipated quantities. Several notions of symmetries that we consider had already been studied for the time-independent case in the literature:
\begin{itemize}
    \item Generalized infinitesimal dynamical symmetries were introduced in \cite{deLeon2020b}, where they were called ``dynamical symmetries''.
    \item (Infinitesimal) dynamical symmetries were introduced in \cite{Gaset2021, Gaset2020a}.
    \item (Infinitesimal) conformal Hamiltonian symmetries are called (infinitesimal) conformal symmetries in \cite{Lainz}.
    \item (Infinitesimal) strict Hamiltonian symmetries were called (infinitesimal) contact symmetries in \cite{Gaset2021} and (infinitesimal) strict symmetries in \cite{Lainz}.
    \item Cartan symmetries were introduced in \cite{deLeon2020b}.
    \item Infinitesimal generalized natural symmetries of the Lagrangian $L$ are called generalized infinitesimal symmetries of $L$ in \cite{deLeon2020b}.
    \item Infinitesimal natural symmetries of the Lagrangian $L$ are called infinitesimal symmetries of $L$ in \cite{deLeon2020b}. These symmetries were also studied in \cite{Gaset2020a}.
    \item Infinitesimal action symmetries are called action symmetries in \cite{Lainz}. This kind of transformations are employed in \cite{de_leon_inverse_2022} to generate equivalent Lagrangians.
\end{itemize}
Some relations of these symmetries with dissipated quantities were also studied in the aforementioned papers. Nevertheless, there was a lack in the literature of a systematic classification of symmetries considering the structures they preserve and the relations between them (see Figures~\ref{fig:infinitesimal_symmetries}, \ref{fig:infinitesimal_symmetries_Lagrangian} and \ref{fig:infinitesimal_symmetries_Lagrangian_Hamiltonian}).

\bigskip
\textbf{Structure of the paper.}
In Section \ref{sec_review}, the most important aspects of cocontact geometry are reviewed. Section \ref{sec_symmetrie_Hamiltonian} is devoted to the study of symmetries and dissipated quantities of time-dependent contact Hamiltonian systems. The symmetries and dissipated quantities of time-dependent contact Lagrangian systems
are discussed in Section \ref{sec_symmetries_Lagrangian}. Some examples are studied in Section \ref{sec_examples}. Finally, Section \ref{sec_conclsuions} provides some conclusions and topics for future research.

\bigskip

\textbf{Notation and conventions.}
Throughout the paper all the manifolds and mappings are assumed to be smooth, connected
and second-countable. Sum over crossed repeated indices is understood.
Given a Cartesian product of manifolds $M_1\times M_2$, the natural projections will be denoted by $\pr{1}\colon M_1\times M_2 \to M_1$ and $\pr{2}\colon M_1\times M_2 \to M_2$, and similarly for a product of $k$ manifolds $M_1\times M_2 \times \cdots \times M_k$.

\section{Review on cocontact mechanics}
\label{sec_review}

In this section the main tools of cocontact geometry are presented. This geometric framework is used to develop a geometric formulation of time-dependent contact systems both in the Hamiltonian and the Lagrangian formalisms. See \cite{deLeon2022} for details.

\subsection{Contact and Jacobi geometry}

First, let us briefly recall the basic notions of contact and Jacobi manifolds that will be employed. For more details see \cite{Geiges2008, deLeon2019, Libermann1987}.

\begin{dfn}
    A \emph{Jacobi manifold} $(M, \Lambda, E)$ is a triple where $M$ is a manifold, $\Lambda$ is a bivector field and $E$ is a vector field on $M$ such that
    \begin{equation}
        [\Lambda, E] = 0, \qquad [\Lambda, \Lambda] = 2 E \wedge \Lambda\, ,
    \end{equation}
    where $[\cdot, \cdot]$ denotes the Schouten--Nijenhuis bracket.
    The pair $(\Lambda, E)$ is called a \emph{Jacobi structure} on $M$.
The \emph{Jacobi bracket} is the map $\{\cdot, \cdot\}\colon \Cinfty(M) \times  \Cinfty(M) \to \Cinfty(M)$ given by
\begin{equation}
    \{f, g\} = \Lambda(\dd f, \dd g) + f E(g) - g E(f).
\end{equation}
\end{dfn}
This bracket is bilinear and satisfies the Jacobi identity. However, unlike Poisson brackets, in general Jacobi brackets do not satisfy the Leibniz rule.

\begin{dfn}
    A (co-oriented) \emph{contact manifold} is a pair $(M, \eta)$ where $M$ is a $(2n+1)$-manifold, and $\eta$ is a one-form on $M$ such that $\eta\wedge(\d\eta)^n$ is a volume form on $M$. The one-form $\eta$ is called a \emph{contact form} on $M$.
\end{dfn}

Given a contact manifold $(M,\eta)$, one can define an isomorphism of $\Cinfty(M)$-modules given by
$$\flat\colon \X(M)\ni X\longmapsto  \contr{X}\d\eta + (\contr{X}\eta)\eta\in \Omega^1(M)\,. $$
Every contact manifold has a unique \emph{Reeb vector field} $\Reeb$, given by $\Reeb=\flat^{-1}(\eta)$. Moreover, to each function $f\in \Cinfty(M)$ one can associate a (contact) \emph{Hamiltonian vector field} $X_f$ given by $\flat(X_f)=\d f-\left(\Reeb f+f\right)\eta$.

Additionally, given a contact manifold $(M,\eta)$, around every point $p\in M$ there exist local coordinates $(q^i, p_i, z)$ such that
$$\eta = \d z - p_i\d q^i\,, \qquad
\Reeb = \parder{}{z}\,, \qquad
X_f = \parder{f}{p_i}\parder{}{q^i} - \left(\parder{f}{q^i} + p_i\parder{f}{z}\right)\parder{}{p_i} + \left(p_i\parder{f}{p_i} - f\right)\parder{}{z}\,.
$$
These coordinates are called \emph{canonical} or \emph{Darboux coordinates}.

A \emph{contact Hamiltonian system} is a triple $(M, \eta, H)$, where $(M, \eta)$ is a contact manifold and $H\in \Cinfty(M)$ is the Hamiltonian function. Its dynamics is given by $X_H$, the Hamiltonian vector field of $H$. There is also a Lagrangian formalism for time-independent contact systems (see \cite{deLeon2019a}).

A contact manifold $(M, \eta)$ has a Jacobi structure $(\Lambda, E)$, where $E= - \Reeb$ and the bivector $\Lambda$ is given by $\Lambda(\alpha, \beta) = - \dd \eta ( \flat^{-1} (\alpha), \flat^{-1} (\beta) )$. The Jacobi bracket $\{\cdot,\cdot\}\colon\Cinfty(M) \times \Cinfty(M) \to \Cinfty(M)$ is
\begin{equation}
    \left\{f,g\right\} 
    = -\d \eta \left(\flat^{-1} \d f, \flat^{-1} \d g\right) - f \Reeb(g) + g \Reeb(f)\,.
\end{equation}

\subsection{Cocontact geometry}

\begin{dfn}
    A \emph{cocontact manifold} is a triple $(M,\tau,\eta)$ where $M$ is a $(2n+2)$-manifold, and $\tau$ and $\eta$ are one-forms on $M$ such that $\d\tau = 0$ and $\tau\wedge\eta\wedge(\d\eta)^n$ is a volume form on $M$. The pair $(\tau, \eta)$ is called a \emph{cocontact structure} on $M$.
\end{dfn}

Given an $n$-dimensional smooth manifold $Q$ with coordinates $(q^i)$ and its cotangent bundle $\cT Q$ with adapted coordinates $(q^i,p_i)$, consider the product manifolds $\R\times\cT Q$, $\cT Q\times\R$ and $\R\times\cT Q\times\R$ with adapted coordinates $(t, q^i, p_i)$, $(q^i,p_i,z)$ and $(t, q^i, p_i, z)$ respectively. The following diagram illustrates this situation and provides some canonical projections:
\begin{center}
    \begin{tikzcd}
        & \R\times\cT Q\times\R \arrow[dl, swap, "\rho_1"] \arrow[dr, "\rho_2"] \arrow[dd, "\pi"] & \\
        \R\times\cT Q \arrow[dr, swap, "\pi_2"] & & \cT Q\times\R \arrow[dl, "\pi_1"] \\
        & \cT Q &
    \end{tikzcd}
\end{center}
Denote by $\theta\in\Omega^1(\R\times\cT Q\times\R)$ the pull-back of the canonical Liouville one-form of the cotangent bundle by the projection $\pi$ given in the diagram above. Hence, $(\tau = \d t,\eta = \d z - \theta)$ is a cocontact structure on the product manifold $\R\times\cT Q\times\R$. This example, also known as \emph{canonical cocontact manifold}, is just a particular case of the following.

\begin{exmpl}\rm\label{ex:R-times-contact}
    Let $(P,\eta_0)$ be a contact manifold and consider the product manifold $M = \R\times P$. Denoting by $\d t$ the pullback to $M$ of the volume form in $\R$ and denoting by $\eta$ the pullback of $\eta_0$ to $M$, we have that $(M, \d t, \eta)$ is a cocontact manifold.
\end{exmpl}

Given a cocontact manifold $(M,\tau,\eta)$, one can define an isomorphism of $\Cinfty(M)$-modules given by
$$ \flat(X)\colon \X(M)\ni X\longmapsto (\contr{X}\tau)\tau + \contr{X}\d\eta + (\contr{X}\eta)\eta\in \Omega^1(M)\,. $$

In addition, every cocontact manifold has two distinguished vector fields $\Rt$ and $\Rz$, characterized by the conditions
$$
    \begin{cases}
        \contr{\Rt}\tau = 1\,,\\
        \contr{\Rt}\eta = 0\,,\\
        \contr{\Rt}\d\eta = 0\,,
    \end{cases}
    \qquad
    \begin{cases}
        \contr{\Rz}\tau = 0\,,\\
        \contr{\Rz}\eta = 1\,,\\
        \contr{\Rz}\d\eta = 0\,,
    \end{cases}
$$
or equivalently, $\Rt = \flat^{-1}(\tau)$ and $\Rz = \flat^{-1}(\eta)$. The vector fields $\Rt$ and $\Rz$ are called \emph{time and contact Reeb vector fields}, respectively.


A cocontact manifold $(M, \tau, \eta)$ is a Jacobi manifold $(M, \Lambda, E)$, where $E= - \Rz$ and the bivector $\Lambda$ is given by $\Lambda(\alpha, \beta) = - \dd \eta ( \flat^{-1} (\alpha), \flat^{-1} (\beta) )$. The Jacobi bracket $\{\cdot,\cdot\}\colon\Cinfty(M) \times \Cinfty(M) \to \Cinfty(M)$ is
\begin{equation}\label{eq:Jacobi_bracket_cocontact}
    \left\{f,g\right\} 
    = -\d \eta \left(\flat^{-1} \d f, \flat^{-1} \d g\right) - f \Rz(g) + g \Rz(f)\,.
\end{equation}

Moreover, given a cocontact manifold $(M,\tau,\eta)$, around every point $p\in M$ there exists a local chart $(U; t, q^i, p_i, z)$ of \emph{canonical} or \emph{Darboux coordinates} such that
$$ \restr{\tau}{U} = \d t\,, \qquad \restr{\eta}{U} = \d z - p_i\d q^i\,,\qquad \restr{\Rt}{U} = \parder{}{t}\,,\qquad\restr{\Rz}{U} = \parder{}{z}\,.$$

\subsection{Hamiltonian formalism}

\begin{dfn}
    A \emph{cocontact Hamiltonian system} is tuple $(M,\tau,\eta,H)$, where $(M,\tau,\eta)$ is a cocontact manifold and $H\in\Cinfty(M)$ is a Hamiltonian function. The \emph{cocontact Hamiltonian equations for a curve} $\psi\colon I\subset \R\to M$ are
    \begin{equation}\label{eq:Ham-eq-cocontact-sections}
            \contr{\psi'}\d \eta = \big(\d H-\Rz(H)\eta-\Rt(H)\tau\big) \circ \psi\,,\qquad
			\contr{\psi'}\eta = -H\circ \psi \,,\qquad
			\contr{\psi'}\tau = 1\,,
    \end{equation}
    where $\psi'\colon I\subset\R\to\T M$ is the canonical lift of the curve $\psi$ to the tangent bundle $\T M$. The \emph{cocontact Hamiltonian equations for a vector field} $X\in\X(M)$ are
    \begin{equation}\label{eq:Ham-eq-cocontact-vectorfields}
        \contr{X}\d \eta = \d H-\Rz(H)\eta-\Rt(H)\tau\,,\qquad
	\contr{X}\eta = -H\,,\qquad
	\contr{X}\tau = 1\,,
    \end{equation}
    which can also be written as 
    $\flat(X)=\d H-\left(\Rz H+H\right)\eta+\left(1-\Rt H\right)\tau$ or
    \begin{equation}\label{eq:Ham-eq-cocontact-vectorfields-equivalent}
        \Lie_X \eta = -\Rz(H)\eta-\Rt(H)\tau\,,\qquad
    	\contr{X}\eta = -H\,,\qquad
    	\contr{X}\tau = 1\,.
    \end{equation}
    These equations have a unique solution called the \emph{cocontact Hamiltonian vector field} $X\equiv X_H$.
\end{dfn}

Given a curve $\psi:I\subset\R\to M$ with local expression $\psi(r)=(f(r),q^i(r),p_i(r),z(r))$, the third equation in \eqref{eq:Ham-eq-cocontact-sections} imposes that $f(r)=r + c$ for some constant $c$, thus we will denote $r\equiv t$, while the other equations read
\begin{equation}\label{eq:Hamilton-cocontact}
   \begin{dcases}
		\dot q^i =\frac{\partial H}{\partial p_i}\,,
		\\
	   \dot p_i = -\left(\frac{\partial H}{\partial q^i}+p_i\frac{\partial H}{\partial z}\right)\,,
		\\
	   \dot z = p_ i\frac{\partial H}{\partial p_i}-H\,.		
	\end{dcases}	
\end{equation}
On the other hand, the local expression of the cocontact Hamiltonian vector field in Darboux coordinates is
$$ X_H = \parder{}{t} + \parder{H}{p_i}\parder{}{q^i} - \left(\parder{H}{q^i} + p_i\parder{H}{z}\right)\parder{}{p_i} + \left(p_i\parder{H}{p_i} - H\right)\parder{}{z}\,. $$
{Note that the integral curves of this vector field satisfy the system of differential equations~\eqref{eq:Hamilton-cocontact}.}

\subsection{Lagrangian formalism}\label{sec:lag}

Given a smooth $n$-dimensional manifold $Q$, consider the product manifold $\R\times\T Q\times \R$ equipped with adapted coordinates $(t, q^i, v^i, z)$. We have the canonical projections
\begin{align*}
	\tau_1&\colon \R\times\T Q\times\R\to\R\ ,&& \tau_1(t, v_q, z) = t\,,\\
	\tau_2&\colon\R\times\T Q\times\R\to\T Q\ ,&& \tau_2(t, v_q, z) = v_q\,,\\
	\tau_3&\colon\R\times\T Q\times\R\to\R\ ,&& \tau_3(t, v_q, z) = z\,,\\
	\tau_0&\colon\R\times\T Q\times\R\to \R\times Q\times\R\ ,&& \tau_0(t, v_q, z) = (t, q, z)\,, 
\end{align*}
which are summarized in the following diagram:
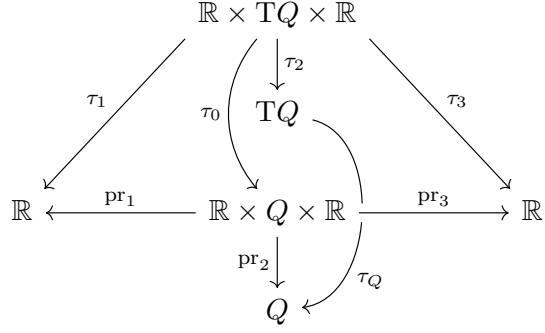
\begin{center}
    \begin{tikzcd}
        && \R\times\T Q\times\R \arrow[ddll, swap, "\tau_1", start anchor={south west}] \arrow[ddrr, "\tau_3", start anchor={south east}] \arrow[d, "\tau_2"] \arrow[dd, bend right=40, swap, "\tau_0"] && \\
        && \T Q  \arrow[dd, swap, bend left=80, pos=0.7, swap, "\tau_Q"] && \\
        \R && \R\times Q\times\R \arrow[ll, swap, "\mathrm{pr}_1"] \arrow[d, swap, "\mathrm{pr}_2"] \arrow[rr, crossing over, "\mathrm{pr}_3"] && \R \\
        && Q 
    \end{tikzcd}
\end{center}

The usual geometric structures of the tangent bundle can be naturally extended to the cocontact Lagrangian phase space $\R\times\T Q\times\R$. In particular, the vertical endomorphism of $\T(\T Q)$ yields a \emph{vertical endomorphism} $\mathcal{S}\colon\T(\R\times\T Q\times\R)\to\T(\R\times\T Q\times\R)$. In the same way, the Liouville vector field on the fiber bundle $\T Q$ gives a \emph{Liouville vector field} $\Delta\in\X(\R\times\T Q\times\R)$. The local expressions of these objects in Darboux coordinates are
\begin{equation}\label{eq:structures_tangent}
    {\cal S} = \frac{\partial}{\partial v^i} \otimes \d q^i \,,\quad \Delta =  v^i\, \frac{\partial}{\partial v^i} \,.
\end{equation}
$$  $$

    Given a path ${\bf c} \colon\R \rightarrow \R\times Q\times\R$ with ${\bf c} = (\mathbf{c}_1,\mathbf{c}_2,\mathbf{c}_3)$, the \emph{prolongation} of ${\bf c}$ to $\R\times\T Q\times\R$ is the path ${\bf \widetilde c} =  (\mathbf{c}_1,\mathbf{c}_2',\mathbf{c}_3) \colon \R \longrightarrow \R\times\T Q \times \R$, where $\mathbf{c}_2'$ is the velocity of~$\mathbf{c}_2$. Every path ${\bf \widetilde c}$ which is the prolongation of a path ${\bf c} \colon\R \rightarrow \R\times Q\times\R$ is called \emph{holonomic}. A vector field 
    $\Gamma \in \X(\R\times\T Q \times \R)$ satisfies the \emph{second-order condition} (it is a {\sc sode}) if all of its integral curves are holonomic.

The vector fields satisfying the second-order condition can be characterized by means of the canonical structures $\Delta$ and $\mathcal{S}$ introduced above, since $X$ is a {\sc sode} if and only if ${\cal S}( \Gamma) = \Delta$.


    A \emph{Lagrangian function} is a function $\L\in\Cinfty(\R\times\T Q\times\R)$. The \emph{Lagrangian energy} associated to $\L$ is the function $E_\L = \Delta(\L)-\L$. The \emph{Cartan forms} associated to $\L$ are
    \begin{equation}\label{eq:thetaL}
        \theta_\L = \transp{\cal S}\circ\d\L\in \Omega^1(\R\times\T Q\times\R)\,,\quad\omega_\L = -\d\theta_\L\in \Omega^2(\R\times\T Q\times\R) \,,
    \end{equation}
    where $\transp{\cal S}$ denotes the transpose operator of the vertical endomorphism.
    The \emph{contact Lagrangian form} is
    $$ \eta_\L=\d z-\theta_\L\in\Omega^1(\R\times\T Q\times\R)\,. $$
    Notice that $\d\eta_\L=\omega_\L$. The couple $(\R\times\T Q\times\R,\L)$ is a \emph{cocontact Lagrangian system}.
The local expressions of these objects are
\begin{gather*}
	E_\L = v^i\parder{\L}{v^i} - \L\,,\qquad \eta_\L = \d z - \frac{\partial\L}{\partial v^i} \,\d q^i \,,\\
	\d\eta_\L = -\parderr{\L}{t}{v^i}\d t\wedge\d q^i -\parderr{\L}{q^j}{v^i}\d q^j\wedge\d q^i -\parderr{\L}{v^j}{v^i}\d v^j\wedge\d q^i -\parderr{\L}{z}{v^i}\d z\wedge\d q^i \,.
\end{gather*}

Not all cocontact Lagrangian systems $(\R\times\T Q\times\R,\L)$ result in the tuple $(\R\times\T Q\times\R,\tau=\d t,\eta_\L, E_\L)$ being a cocontact Hamiltonian system because the condition $\tau\wedge\eta\wedge(\d\eta_\L)^n\neq 0$ is not always fulfilled. The Legendre map characterizes the Lagrangian functions that will result in cocontact Hamiltonian systems.

    Given a Lagrangian function $\L \in\Cinfty(\R\times \T Q\times\R)$, the \emph{Legendre map} associated to $\L$ is its fiber derivative \cite{Gracia2000}, considered as a function on the vector bundle $\tau_0 \colon\R\times \T Q\times\R \to \R\times Q \times \R$; that is, the map $\F\L \colon \R\times\T Q \times\R \to \R\times\cT Q \times \R$ with local expression
    $$
        \F\L (t,v_q,z) = \left( t,\F\L(t,\cdot,z) (v_q),z \right)\,,
    $$
    where $\F\L(t,\cdot,z)$ is the usual Legendre map associated to the Lagrangian $\L(t,\cdot,z)\colon\T Q\to\R$ with the variables $t$ and $z$ fixed.

The Cartan forms can also be defined as $\theta_\L={\cal F}L^{\;*}(\pi^*\theta_0)$ and $\omega_\L={\cal F}L^{\;*}(\pi^*\omega_0)$, where $\theta_0$ and $\omega_0 = -\d\theta_0$ are the canonical one- and two-forms of the cotangent bundle and $\pi$ is the natural projection $\pi\colon\R\times\cT Q\times\R\to\cT Q$.

\begin{prop}\label{prop:regular-Lagrangian}
    Given a Lagrangian function $\L$ the following statements are equivalent:
    \begin{enumerate}[{\rm(1)}]
        \item The Legendre map $\F \L$ is a local diffeomorphism.
        \item The fiber Hessian $\F^2\L \colon\R\times\T Q\times\R \longrightarrow (\R\times\cT Q\times\R)\otimes (\R\times\cT Q\times\R)$ of $\L$ is everywhere nondegenerate (the tensor product is understood to be of vector bundles over $\R\times Q \times \R$).
        \item The triple $(\R\times\T Q\times\R,\d t, \eta_\L)$ is a cocontact manifold.
    \end{enumerate}
\end{prop}

A Lagrangian function $\L$ is \emph{regular} if the equivalent statements in the previous proposition hold. Otherwise $\L$ is \emph{singular}. Moreover, $\L$ is \emph{hyperregular} if $\F\L$ is a global diffeomorphism. Thus, every \emph{regular} cocontact Lagrangian system yields the cocontact Hamiltonian system $(\R\times\T Q\times\R,\d t, \eta_\L, E_\L)$.

The local expressions of the Reeb vector fields are
\begin{equation*}
    \Rt^\L = \parder{}{t} - W^{ij}\parderr{\L}{t}{v^j}\parder{}{v^i}\,,\qquad
    \Rz^\L = \parder{}{z} - W^{ij}\parderr{\L}{z}{v^j}\parder{}{v^i}\,,
\end{equation*}
where $(W^{ij})$ is the inverse of the Hessian matrix of the Lagrangian $\L$, namely $W^{ij}W_{jk}=\delta^i_k$.

If the Lagrangian $\L$ is singular, the Reeb vector fields are not uniquely determined, actually, they may not even exist \cite{deLeon2022}.

\subsubsection{The Herglotz--Euler--Lagrange equations}

\begin{dfn}\label{dfn:Euler-Lagrange-equations}
	Given a regular cocontact Lagrangian system $(\R\times\T Q\times\R,\L)$ the \emph{Herglotz--Euler--Lagrange equations} for a holonomic curve ${\bf\widetilde c}\colon I\subset\R \to\R\times\T Q\times\R$ are
	\begin{equation}\label{eq:Euler-Lagrange-curve}
		\begin{dcases}
			\contrp{{\bf\widetilde c}'}\d\eta_\L = \left(\d E_\L - {\Rt^\L}(E_\L)\d t - {\Rz^\L}(E_\L)\eta_\L\right)\circ{\bf\widetilde c}\,,\\
			\contrp{{\bf\widetilde c}'}\eta_\L = - E_\L\circ{\bf\widetilde c}\,,\\
			\contrp{{\bf\widetilde c}'}\d t = 1\,,
		\end{dcases}
	\end{equation}
	where ${\bf\widetilde c}'\colon I\subset\R\to\T(\R\times\T Q\times\R)$ is the canonical lift of ${\bf\widetilde c}$ to $\T(\R\times\T Q\times\R)$. The \emph{cocontact Lagrangian equations} for a vector field $X_\L\in\X(\R\times\T Q\times\R)$ are 
	\begin{equation}\label{eq:Euler-Lagrange-field}
		\begin{cases}
			\contr{X_\L}\d\eta_\L = \d E_\L-\Rt^\L(E_\L)\d t-\Rz^\L(E_\L)\eta_\L\,,\\
			\contr{X_\L}\eta_\L = -E_\L \,,\\
			\contr{X_\L}\d t = 1\,.
		\end{cases}
	\end{equation}
	The only vector field solution to these equations is the \emph{cocontact Lagrangian vector field}.
\end{dfn}

Equations~\eqref{eq:Euler-Lagrange-curve} and \eqref{eq:Euler-Lagrange-field} are the Lagrangian counterparts of equations \eqref{eq:Ham-eq-cocontact-sections} and \eqref{eq:Ham-eq-cocontact-vectorfields}, respectively.
    The cocontact Lagrangian vector field of a regular cocontact Lagrangian system $(\R\times\T Q\times\R,\L)$ coincides with the cocontact Hamiltonian vector field of the cocontact Hamiltonian system $(\R\times\T Q\times\R,\d t,\eta_\L,E_\L)$.


\begin{thm}\label{thm:regular-lagrangian}
    If $\L$ is a regular Lagrangian, then $X_\L\equiv\Gamma_\L$ is a {\sc sode}, called the \emph{Herglotz--Euler--Lagrange vector field} for the Lagrangian $\L$.
\end{thm}

The coordinate expression of the Herglotz--Euler--Lagrange vector field is
\begin{equation}\label{eq:Euler-Lagrange-field_local}
    \Gamma_\L = \parder{}{t} + v^i\parder{}{q^i} + W^{ji}\left( \parder{\L}{q^j} - \parderr{\L}{t}{v^j} - v^k\parderr{\L}{q^k}{v^j} - \L\parderr{\L}{z}{v^j} + \parder{\L}{z}\parder{\L}{v^j} \right)\parder{}{v^i} + \L\parder{}{z}\,.
\end{equation}
An integral curve of $\Gamma_\L$ fulfills the \emph{Herglotz--Euler--Lagrange equations} for dissipative systems:
$$ \frac{\d}{\d t}\left(\parder{\L}{v^i}\right) - \parder{\L}{q^i} =\parder{\L}{z}\parder{\L}{v^i}\,,\qquad
\dot z = \L\,.
$$
These equations can also be obtained variationally from the Herglotz principle \cite{Herglotz1930} (see also \cite{deLeon2019a}).
{Roughly speaking, the variable $z$ can be interpreted as the action of the Lagrangian system.}

\section{Symmetries and dissipated quantities of cocontact Hamiltonian systems}\label{section_symmetries_Hamiltonian}
\label{sec_symmetrie_Hamiltonian}


In this section we will study the symmetries of regular time-dependent contact mechanical systems and their associated conserved and dissipated quantities. A summary of the symmetries and their relations can be found in Figure~\ref{fig:infinitesimal_symmetries}. In some cases we will restrict ourselves to the case of cocontact manifolds of the form $M = \R\times N$ where $N$ is a contact manifold (see Example \ref{ex:R-times-contact}). In this case, the natural projection $\R\times N\to\R$ defines a global canonical coordinate $t$ on the cocontact manifold $\R\times N$.

\begin{dfn}\label{dfn:cocontactomorphisms}
    Let $(M, \tau, \eta)$ be a cocontact manifold. A diffeomorphism $\Phi\colon M \to M$ is called a \emph{conformal cocontatomorphism} if $\Phi^\ast \tau = \tau$ and $\Phi^\ast \eta = f \eta$ for some non-vanishing function $f$ on $M$ called the \emph{conformal factor}.
    A \emph{(strict) cocontactomorphism} is a conformal cocontactomorphism with conformal factor $f\equiv 1$.

    An \emph{infinitesimal conformal (resp.~strict) cocontactomorphism} is a vector field $Y\in \mathfrak{X}(M)$ whose flow is a one-parameter group of conformal (resp.~strict) cocontactomorphisms.
\end{dfn}

\begin{prop}
    Let $\Phi:M\to M$ be a cocontactomorphism (i.e., $\Phi^\ast\eta = \eta$ and $\Phi^\ast \tau = \tau$), then $\Phi$ preserves the Reeb vector fields (i.e., $\Phi_\ast \Rt= \Rt$ and $\Phi_\ast \Rz= \Rz$).
\end{prop}
\begin{proof}
   Suppose that $\Phi$ is a cocontactomorphism. We have
   \begin{align*}
		& \contrp{\Phi_\ast^{-1}\Rt}(\Phi^\ast\d\eta) = \Phi^\ast(\contr{\Rt}\d\eta) = 0\,,\\
		& \contrp{\Phi_\ast^{-1}\Rt}(\Phi^\ast\tau) = \Phi^\ast(\contr{\Rt}\tau) = 1\,,\\
		& \contrp{\Phi_\ast^{-1}\Rt}(\Phi^\ast\eta) = \Phi^\ast(\contr{\Rt}\eta) = 0\,.
	\end{align*}
	Since $\Phi^\ast\eta = \eta$ and $\Phi^\ast\tau = \tau$, by the uniqueness of the time Reeb vector field, we get that $\Phi_\ast \Rt = \Rt$. Analogously, one can see that the contact Reeb vector field is also preserved. 
\end{proof}

\begin{cor}
    If a vector field $Y\in \mathfrak{X}(M)$ is an infinitesimal cocontactomorphism (i.e., $\liedv{Y}\eta= \liedv{Y} \tau =0$), then $[Y, \Rt]=[Y, \Rz]=0$. 
\end{cor}

It is worth noting that the converse is false. 

\begin{exmpl}
    Consider the cocontact manifold $(M, \tau, \eta)$ where $M=\R^4,\ \tau=\d t$ and $\eta=\d z-p\d q$, where $(t, q, p, z)$ are canonical coordinates. Clearly, the vector field $Y=\tparder{}{p}$ on $M$ preserves the Reeb vector fields $\Rt = \tparder{}{t}$ and $\Rz = \tparder{}{z}$. However, it is not an infinitesimal cocontactomorphism. Indeed,
    \begin{equation}
        \liedv{Y} \eta = \contr{Y} \d \eta = - \dd q \neq 0.
    \end{equation}
    Similarly, one can check that the map $\Phi\colon M \to M,\  (t, q, p, z)\mapsto (t, q, 2p, z)$ is a diffeomorphism preserving the Reeb vector field, but it is not a cocontactomorphism
\end{exmpl}

\subsection{Dissipated and conserved quantities of cocontact systems}

\begin{dfn}
    Let $(M,\tau,\eta,H)$ be a cocontact Hamiltonian system. A \emph{dissipated quantity} is a function $f\in\Cinfty(M)$ such that
    $$ X_H(f) = -\Rz(H) f\,. $$
\end{dfn}

Notice that, unlike in the time-independent contact case, the Hamiltonian function is not a dissipated quantity. Taking into account that
$$ X_H(H) = -\Rz(H)H + \Rt(H)\,, $$
it is clear that $H$ is a dissipated quantity if it is time-independent, namely $\Rt(H) = 0$. This resembles the cosymplectic case, where the Hamiltonian function is conserved if, and only if, it is time-independent.

\begin{prop} \label{prop_bracket_dissipated}
    Let $(M,\tau,\eta,H)$ be a cocontact Hamiltonian system. A function $f\in \Cinfty(M)$ is a dissipated quantity if and only if $\{f, H\} = \Rt(f)$, where $\{\cdot,\cdot\}$ is the Jacobi bracket associated to the cocontact structure $(\tau,\eta)$.
\end{prop}

\begin{proof}
    The Jacobi bracket of $f$ and $H$ is given by equation~\eqref{eq:Jacobi_bracket_cocontact}:
    \begin{equation}
    \begin{aligned}
        \left\{f,H\right\} 
        = -\d \eta \left(\flat^{-1} \d f, \flat^{-1} \d H\right) - f \Rz(H) + H \Rz(f),
    \end{aligned}
    \end{equation}
    but
    \begin{equation}
        \flat^{-1} \d f = X_f + \left(\Rz(f) + f\right) \Rz - \left(1 - \Rt(f) \right) \Rt,
    \end{equation}
    so, taking into account equations \eqref{eq:Ham-eq-cocontact-vectorfields},
    \begin{equation}
        \contr{\flat^{-1} \d f} \d \eta = \contr{X_f} \d \eta
        = \d f - \Rz(f) \eta -\Rt(f) \tau\, ,
    \end{equation}
    and thus
    \begin{equation}
        \d \eta (\flat^{-1} \d f, \flat^{-1} \d H) 
        = X_H(f)  + \Rs(f) H - \Rt(f)
        \,.
    \end{equation}
    Hence,
    \begin{equation}
        \{f, H\} = -X_H(f) - \Rz (H) f + \Rt(f)\, ,
    \end{equation}
    so
    \begin{equation}
         \left\{H,f\right\} + \Rt(f) = X_H(f)  + \Rs(H) f
         \, .
    \end{equation}
    In particular, the right-hand side vanishes if and only if $f$ is a dissipated quantity.
\end{proof}

The symmetries that we shall present yield dissipated quantities. However, we are also interested in finding conserved quantities.

\begin{dfn}
    A \emph{conserved quantity} of a cocontact Hamiltonian system $(M,\tau,\eta,H)$ is a function $g\in\Cinfty(M)$ such that
    $$ X_H (g) = 0\,. $$
\end{dfn}

Taking into account that every dissipated quantity changes with the same rate $\Rz(H)$, we have the following result, whose proof is straightforward.

\begin{prop}
    Consider a cocontact Hamiltonian system $(M,\tau,\eta,H)$. Then
    \begin{enumerate}[{\rm(1)}]
        \item if $f_1$ and $f_2$ are dissipated quantities and $f_2\neq 0$, then $f_1/f_2$ is a conserved quantity,
        \item if $f$ is a dissipated quantity and $g$ is a conserved quantity, then $fg$ is a dissipated quantity,
        \item if $f_1$ and $f_2$ are dissipated quantities, $a_1 f_1 + a_2 f_2$ is also a dissipated quantity for any $a_1, a_2 \in \R$,
        \item if $g_1$ and $g_2$ are conserved quantities, $a_1 g_1 + a_2 g_2 + a_3$ is also a conserved quantity for any $a_1, a_2, a_3\in \R$.
    \end{enumerate}
\end{prop}



\subsection{Generalized infinitesimal dynamical symmetries}
The following result motivates the definition of the most general type of symmetries with associated dissipated quantities.


\begin{thm}[Noether's theorem]\label{Noether_thm}
    Consider the cocontact Hamiltonian system $(M,\tau,\eta,H)$. Let $Y\in \mathfrak{X}(M)$.
    If $\eta([Y,X_H]) = 0$ and $\contr{Y}\tau = 0$, then $f = -\contr{Y}\eta$ is a dissipated quantity. Conversely, given a dissipated quantity $f$, the vector field $Y = X_f - \Rt$, where $X_f$ is the Hamiltonian vector field associated to $f$, 
    verifies $\eta([Y,X_H]) = 0$, $\contr{Y}\tau = 0$ and $f = -\contr{Y}\eta$.
\end{thm}

\begin{proof}
    Let $f = -\contr{Y}\eta$, where $Y$ satisfies $\eta([Y,X_H]) = 0$ and $\contr{Y}\tau = 0$.
     Then,
    \begin{align*}
        \Lie_{X_H}f &= -\Lie_{X_H}\contr{Y}\eta
        = -\contr{Y}\Lie_{X_H}\eta - \contr{[X_H,Y]}\eta
        = \contr{Y}\left( \Rz(H)\eta + \Rt(H)\tau \right)
        \\&
        = \Rz(H)\contr{Y}\eta
        = -\Rz(H) f\,,
    \end{align*}
    and thus $f$ is a dissipated quantity.

    On the other hand, given a dissipated quantity $f$, let $Y = X_f - \Rt$. Then, it is clear that $f = -\contr{Y}\eta$. In addition, $\contr{Y}\tau = 0$, and
    \begin{align}
        \contr{[X_H,Y]}\eta 
        &= \liedv{X_H} \contr{Y} \eta -  \contr{Y} \liedv{X_H} \eta
        = -\liedv{X_H} f  + \contr{Y} \left(\Rz(H) \eta +\Rt(H) \tau\right)
        \\
        &= \Rz(H) f - \Rz(H) \contr{Y} \eta = 0\,,
    \end{align}
    where we have used equations \eqref{eq:Ham-eq-cocontact-vectorfields-equivalent}.
\end{proof}

This result motivates the following definition.

\begin{dfn}
    Let $(M,\tau,\eta,H)$ be a cocontact Hamiltonian system and let $X_H$ be its cocontact Hamiltonian vector field. A \emph{generalized infinitesimal dynamical symmetry} is a vector field $Y\in\X(M)$ such that $\eta([Y,X_H]) = 0$ and $\contr{Y}\tau = 0$.
\end{dfn}

In particular, if $H$ is a time-independent Hamiltonian function, then $H$ is a dissipated quantity and its associated generalized infinitesimal dynamical symmetry is the Hamiltonian vector field $X_H$.

{Theorems 3 and 4 of \cite{Bravetti2021} are the analogous of Theorem \ref{Noether_thm} in the extended contact phase space (instead of the cocontact) formalism.}

\begin{rmrk}
Despite the condition $\tau(Y)=0$, the dissipated quantity associated to a generalized infinitesimal dynamical symmetry $Y$ may be time-dependent. Indeed,
\begin{equation}
    \liedv{\Rt} f = - \liedv{\Rt} \contr{Y} \eta = -\contr{[\Rt, Y]} \eta - \contr{Y} \liedv{\Rt} \eta
    = -\eta\big([\Rt, Y] \big)
    = - \parder{Y^z}{t} + p_i\parder{Y^{q^i}}{t} ,
\end{equation}
where $Y=Y^{q^i}\tparder{}{q^i}+Y^{v^i}\tparder{}{v^i}+Y^{z}\tparder{}{z}$.

\end{rmrk}




\subsection{Other symmetries}
We are now interested in other types of symmetries which preserve more properties of the system, such as the dynamical vector field or the Hamiltonian function.

\begin{dfn}
    Let $(M,\tau,\eta,H)$ be a cocontact Hamiltonian system and let $X_H$ be its cocontact Hamiltonian vector field.
    \begin{enumerate}[{\rm(1)}]
        \item If $M = \R\times N$ with $N$ a contact manifold, a \emph{dynamical symmetry} is a diffeomorphism $\Phi\colon M\to M$ such that $\Phi_\ast X_H = X_H$ and $\Phi^\ast t = t$.
        \item An \emph{infinitesimal dynamical symmetry} is a vector field $Y\in\X(M)$ such that $\Lie_Y X_H = [Y,X_H] = 0$ and $\contr{Y}\tau = 0$. In particular, if $M = \R\times N$, the flow of $Y$ is made of dynamical symmetries.
    \end{enumerate}
\end{dfn}

Generalized infinitesimal dynamical symmetries receive that name since they satisfy weaker conditions than infinitesimal dynamical symmetries. It is clear that every infinitesimal dynamical symmetry is a generalized infinitesimal dynamical symmetry. We also define a generalization of dynamical symmetries as follows:

\begin{dfn}\label{dfn:gendynsym}
    Let $(M,\tau,\eta,H)$ be a cocontact Hamiltonian system, where $M = \R\times N$ with $N$ a contact manifold, and let $X_H$ be its cocontact Hamiltonian vector field. A \emph{generalized dynamical symmetry} is a diffeomorphism $\Phi\colon M\to M$ such that $\eta(\Phi_\ast X_H) = \eta(X_H)$ and $\Phi^\ast  t = t$.
\end{dfn}

Unlike other symmetries with infinitesimal counterparts, the flow of a generalized infinitesimal dynamical symmetry is not necessarily made of generalized dynamical symmetries.

\begin{exmpl}\label{ex:example}
    Consider the cocontact Hamiltonian system $(\R^4\setminus\{0\}, \tau, \eta, H)$, with $\tau= \dd t,\ \eta = \dd z - p \dd x$ and 
    $$H=\frac{p^2}{2} + z\, ,$$
    where $(t, x, p, z)$ are the canonical coordinates in $\R^4$. 
    The family of diffeomorphisms
    \begin{equation}
    \begin{aligned}
        \Phi^r\colon \R^4\setminus\{0\} &\to \R^4\setminus\{0\}\\
        \left( t, x, p, z\right ) & \mapsto \left( t, x, p+r, z \right)
    \end{aligned}
    \end{equation}
    for $r\in\mathbb{R}$, is generated by the vector field $Y =\parder{}{p}$. 
    One can check that $Y$
    is a generalized infinitesimal dynamical symmetry, but $\Phi^r$ is not a generalized dynamical symmetry for $r\neq 0$. Indeed, for
    \begin{equation}
        X_H = \parder{}{t} + p \parder{}{x} - p \parder{}{p} + \left( \frac{p^2}{2} - z\right) \parder{}{z}\,,
    \end{equation}
    we have
    \begin{equation}
        \Phi^r_*X_H = \parder{}{t} + (p-r) \parder{}{x} - (p-r) \parder{}{p} + \left( \frac{(p-r)^2}{2} - z\right) \parder{}{z}\neq X_H\,,
    \end{equation}
    and $\eta(\Phi^r_*X_H)\neq\eta(X_H)$. 
\end{exmpl}

The (infinitesimal) dynamical symmetries defined above are the counterparts of (infinitesimal) dynamical symmetries in symplectic Hamiltonian systems (see \cite{deLeon1989,Roman-Roy2020} and references therein).
They are of interest since they map trajectories of the system onto other trajectories. As a matter of fact, if $\sigma\colon \R \to M$ is an integral curve of $X_H$ and $\Phi$ is a dynamical symmetry, then $\Phi \circ \sigma$ is also an integral curve of $X_H$.
In addition, we have the following result.

\begin{prop}
    Infinitesimal dynamical symmetries close a Lie subalgebra of $(\mathfrak{X}(M), [\cdot, \cdot])$. In other words, given two infinitesimal dynamical symmetries $Y_1,Y_2\in\X(M)$, its Lie bracket $[Y_1,Y_2]$ is also an infinitesimal dynamical symmetry.

    Moreover, dynamical symmetries form a Lie subgroup of $\Diff(M)$, that is, for any pair of dynamical symmetries $\Phi_1$ and $\Phi_2$, the composition $\Phi_1 \circ \Phi_2$ is also a dynamical symmetry.
\end{prop}
\begin{proof}
    Using the Jacobi identity,
    $$ [[Y_1,Y_2],X_H] = [Y_2,[X_H,Y_1]] + [Y_1,[Y_2,X_H]] = 0\,. $$
    In addition,
    $$ \contr{[Y_1,Y_2]}\tau = \Lie_{Y_1}\contr{Y_2}\tau - \contr{Y_2}\Lie_{Y_1}\tau = -\contr{Y_2}\left( \contr{Y_1}\d\tau + \d \contr{Y_1}\tau \right) = 0\,. $$

    On the other hand, if $\Phi_1$ and $\Phi_2$ are dynamical symmetries, then 
    $$(\Phi_1\circ \Phi_2)_\ast X_H = (\Phi_1)_\ast (\Phi_2)_\ast X_H 
    = (\Phi_1)_\ast X_H = X_H\, ,$$ 
    and $(\Phi_1\circ \Phi_2)^\ast t = \Phi_2^\ast \Phi_1^\ast t = \Phi_2^\ast t = t$. Obviously, $\Phi\equiv \id$ is a dynamical symmetry. Finally, if $\Phi$ is a dynamical symmetry, then
    $$X_H = (\Phi^{-1}\circ \Phi)_\ast X_H = \Phi^{-1}_\ast \Phi_\ast X_H = \Phi^{-1}_\ast X_H\, ,$$
    and similarly $(\Phi^{-1})^\ast t = t$. This proves that dynamical symmetries form a group under composition.
\end{proof}

Generalized infinitesimal dynamical symmetries do not close a Lie algebra, as the counterexample below shows. 

\begin{exmpl}
    Consider the cocontact Hamiltonian system from Example \ref{ex:example}. 
    Given the vector fields
    $$ Y=\frac{\partial}{\partial p}\qquad \text{and}\qquad Z=\frac{x}{2}\frac{\partial}{\partial x}+\frac{p}{2}\frac{\partial}{\partial p}+(z+p)\frac{\partial}{\partial z}\,,$$
    one can check that $Y$ is a generalized infinitesimal dynamical symmetry and $Z$ is an infinitesimal dynamical symmetry. Nevertheless,
    $$ [Y,Z]=\frac12\frac{\partial}{\partial p}+\frac{\partial}{\partial z}$$
    is not a generalized infinitesimal symmetry.
\end{exmpl}

A natural type of objects that conserve the geometry of the system are the (infinitesimal) $f$-conformal cocontactomorphisms (see Definition \ref{dfn:cocontactomorphisms}). Since the function $H$ is independent of the cocontact structure $(\tau, \eta)$, in general $f$-conformal cocontactomorphisms are not generalized dynamical symmetries. The necessary and sufficient condition is shown in the next result.

\begin{prop}
Let $(M,\tau,\eta,H)$ be a cocontact Hamiltonian system.
\begin{enumerate}[{\rm(1)}]
    \item Let $\Phi:M\rightarrow M$ be an $f$-conformal cocontactomorphism of the cocontact
    manifold $(M,\tau,\eta)$, namely $\Phi^*\eta=f\eta$ and $\Phi^*\tau=\tau$. Then, $\eta(\Phi_*X_H)=\eta(X_H)$ if, and only if, $\Phi^*H=fH$. Moreover, for a cocontact Hamiltonian system of the form presented in Definition \ref{dfn:gendynsym}, $\Phi$ is a generalized dynamical symmetry if, and only if, $\Phi^*H=fH$ and $\Phi^*t=t$.
    \item Let $Y\in\mathfrak{X}(M)$ be an infinitesimal $g$-conformal cocontactomorphism of the cocontact manifold $(M,\tau,\eta)$, namely $\Lie_{Y}\eta=g\eta$ and $\Lie_{Y}\tau=0$. Then, $\eta([Y,X_H])=0$ if, and only if, $\Lie_{Y}H=gH$. In particular, $Y$ is a generalized infinitesimal dynamical symmetry if, and only if, $\Lie_{Y}H=gH$ and $\contr{Y}\tau=0$.
\end{enumerate}
\end{prop}

\begin{proof} If $X_H$ is the solution of the cocontact Hamiltonian system $(M,\tau,\eta,H)$, we have that $\contr{X_H}\eta=-H$, so
$$
    \Phi^*H=-\Phi^*(\contr{X_H}\eta)=- \contr{\Phi_*X_H}\Phi^*\eta=-f\contr{\Phi_*X_H}\eta \,.
$$
If $\Phi$ is a generalized dynamical symmetry, then $\contr{\Phi_*X_H}\eta=\contr{X_H}\eta$, and therefore $\Phi^*H=fH$. Conversely, if $\Phi^*H=fH$, then
$$
f\contr{\Phi_*X_H}\eta=-\Phi^*H=-fH=f\contr{X_H}\eta\,.
$$
Since $f\neq0$ everywhere, we conclude that $\contr{\Phi_*X_H}\eta=\contr{X_H}\eta$.

The infinitesimal case is proved with a similar argument using the relation
$$
\Lie_{Y}H=-\Lie_{Y}(\contr{X_H}\eta)=- \contr{[Y,X_H]}\eta-\contr{X_H}\Lie_{Y}\eta=- \contr{[Y,X_H]}\eta-g\contr{X_H}\eta=- \contr{[Y,X_H]}\eta+gH \,.
$$
\end{proof}

This result justifies the following definition.
\begin{dfn}
    Let $(M,\tau,\eta,H)$ be a cocontact Hamiltonian system.
    \begin{enumerate}[{\rm(1)}]
        \item A \emph{$f$-conformal Hamiltonian symmetry} is a diffeomorphism $\Phi:M \to M$ such that
        $$ \Phi^*t = t \ ,\qquad \Phi^\ast\eta = f\eta\ ,\qquad \Phi^\ast H 
        = fH\,,
        $$
        where $f\in\Cinfty(M)$ does not vanish anywhere, $M = \R\times N$ with $(N,\eta)$ a contact manifold, and $t$ is the canonical coordinate of $\R$.  If $\Phi$ is a cocontactomorphism (i.e., if $f \equiv 1$), we say that $\Phi$ is a \emph{strict Hamiltonian symmetry}.
        \item An \emph{infinitesimal $\rho$-conformal Hamiltonian symmetry} is a vector field $Y\in\X(M)$ such that
        $$ \contr{Y}\tau = 0\ ,\qquad \Lie_Y\eta = \rho\eta\ ,\qquad \Lie_Y H 
        = \rho H\,,
        $$
        where $\rho \in\Cinfty(M)$. In particular, if $M = \R\times N$, the flow of $Y$ is made of conformal Hamiltonian symmetries. If $Y$ is an infinitesimal cocontactomorphism (i.e., if $\rho \equiv 0$), $Y$ is said to be an \emph{infinitesimal strict Hamiltonian symmetry}.
    \end{enumerate}
\end{dfn}

These symmetries correspond, in time-independent contact systems, to  ``contact symmetries'' (see \cite{Gaset2020a}). The symplectic counterparts of (infinitesimal) strict Hamiltonian symmetries are sometimes referred to as ``(infinitesimal) Noether symmetries'' (see \cite{Roman-Roy2020} and references therein).


If a conserved quantity is known, (infinitesimal) dynamical symmetries can be used to compute additional conserved quantities. Similarly, if a dissipated quantity is known, (infinitesimal) strict Hamiltonian symmetries can be used to compute new dissipated quantities.

\begin{prop}\label{prop:new_quantities}
Suppose that $g\in\Cinfty(M)$ is a conserved quantity and $f\in\Cinfty(M)$ is a dissipated quantity.
\begin{enumerate}[{\rm(1)}]
    \item If $\Phi\colon M\to M$ is a strict Hamiltonian symmetry and a dynamical symmetry, then $\widehat{f}=f\circ \Phi=\Phi^\ast f$ is also a dissipated quantity.
    \item If $Y \in \mathfrak{X}(M)$ is an infinitesimal strict Hamiltonian symmetry and an infinitesimal dynamical symmetry, then $\widetilde{f}=\liedv{Y} f$ is also a dissipated quantity.
    \item If $\Phi\colon M\to M$ is a dynamical symmetry, then $\widehat{g}=g\circ \Phi=\Phi^\ast g$ is also a conserved quantity.
    \item If $Y\in\X(M)$ is an infinitesimal dynamical symmetry, then $\widetilde{g} = \Lie_Y g$ is also a conserved quantity.
\end{enumerate}
\end{prop}
\begin{proof}
    Let $f$ and $g$ be a dissipated and a conserved quantity, respectively.
    Suppose that $\Phi\colon M\to M$ is an strict Hamiltonian symmetry and a dynamical symmetry. Then,
    $$ \liedv{X_H} \widehat{f}
    = \liedv{X_H} (\Phi^\ast f)
    = \Phi^\ast \left(\liedv{\Phi_\ast X_H} f\right)
    = \Phi^\ast \left(\liedv{X_H} f\right)
    = \Phi^\ast \left(-\liedv{\Rz}(H) f\right)
    = -\liedv{\Rz}(H) \Phi^\ast f\, .
    $$
    Similarly, if $\Phi$ is a dynamical symmetry, then
    $$\liedv{X_H} \widehat{g}
    = \liedv{X_H} (\Phi^\ast g)
    = \Phi^\ast \left(\liedv{\Phi_\ast X_H} g\right)
    = \Phi^\ast \left(\liedv{X_H} g\right)
    = 0\, .
    $$
    If $Y \in \mathfrak{X}(M)$ is an infinitesimal dynamical symmetry, then
     $$ \Lie_{X_H}\widetilde{g} = \Lie_{X_H}\Lie_Y g = \Lie_{[X_H,Y]} g + \Lie_Y\Lie_{X_H} g = 0\,.$$
    Finally, if $Y \in \mathfrak{X}(M)$ is an infinitesimal strict Hamiltonian symmetry and an infinitesimal dynamical symmetry, we have that
     \begin{equation}
         \liedv{X_H} \widetilde{f}
        = \liedv{X_H} \left( \liedv{Y} f \right)
        =\liedv{[X_H,Y]} f + \liedv{Y} \left(\liedv{X_H} f \right)
        = \liedv{Y}\left( - \liedv{\Rz} (H) f\right)
        = - \liedv{\Rz} (H)  \left( \liedv{Y} f \right)\, . \nonumber
    \end{equation}
\end{proof}




The results from Proposition \ref{prop:new_quantities} cannot be extended to generalized infinitesimal dynamical symmetries. As a matter of fact, we have the following counterexample.

\begin{exmpl}\label{counterxample_new_quantities_generalized}
    Consider the same system as in Example~\ref{ex:example}. Let $Y\in \mathfrak{X}(\R^4\setminus \{0\})$ be the vector field $Y=\frac{\partial}{\partial p}$. We have that $[Y,X_H]\neq0$, but $\eta([Y,X_H])=0$ therefore, it is a generalized infinitesimal symmetry but it is not a dynamical symmetry.
    The function $f(t,x,p,z) = p$ is a dissipated quantity, but $\Lie_Y f =1$ is not a dissipated quantity. Likewise, $\Lie_Y H = p$ is not a dissipated quantity either. Finally,
    \begin{equation}
        \Lie_Y\frac{H}{f} =\frac12-\frac{z}{p^2}\,,
    \end{equation}
    is not a conserved quantity.
\end{exmpl}

It is also worth mentioning that preserving the Hamiltonian is not a sufficient condition for a diffeomorphism (vector field) to be a (infinitesimal) dynamical symmetry. It is not a sufficient condition for being a generalized (infinitesimal) dynamical symmetry either.
\begin{exmpl}
    Consider the cocontact Hamiltonian system $(\R^4, \tau, \eta, H)$, with $\tau= \dd t,\ \eta = \dd z - p \dd x$ and 
    $$H=\frac{p^2}{2}\, ,$$
    where $(t, x, p, z)$ are the canonical coordinates in $\R^4$. Its Hamiltonian vector field is given by
    \begin{equation}
        X_H = \parder{}{t} +  p \parder{}{x} + \frac{p^2}{2} \parder{}{z}.
    \end{equation}    
    Let $Y= z\tparder{}{z}$. One can check that $Y(H) = 0$, but $[Y, X_H]\neq 0$ and $\eta([Y, X_H]) \neq 0$. Similarly, $\Phi \colon \R^4 \to \R^4,\ (t, x, p, z)\mapsto (t, x, p, 2 z)$ is a diffeomorphism preserving the Hamiltonian function $H$ but not the vector field $X_H$.  
\end{exmpl}

Furthermore, we can consider the following generalization of infinitesimal $\rho$-conformal Hamiltonian symmetries.

\begin{dfn}
    Given a cocontact Hamiltonian system $(M,\tau,\eta, H)$, a \emph{$(\rho, g)$-Cartan symmetry} is a vector field $Y\in\X(M)$ such that
    $$ \Lie_Y\eta = \rho \eta + \d g\,,\qquad \Lie_Y H = \rho H + g \Rz(H)\,,\qquad \contr{Y}\tau = 0\,, $$
    where $\rho,g\in\Cinfty(M)$.
\end{dfn}

Clearly, a $\rho$-conformal Hamiltonian symmetry is a $(\rho, 0)$-Cartan symmetry. On the other hand, $(0,g)$-Cartan symmetries are the analogous of Cartan symmetries in symplectic Hamiltonian systems (see \cite{Lopez1999} for instance). 

\begin{thm}\label{thm_Cartan_dissipated}
    If $Y$ is a $(\rho, g)$-Cartan symmetry of a cocontact Hamiltonian system $(M,\tau,\eta,H)$, the function $f = g - \contr{Y}\eta$ is a dissipated quantity.
\end{thm}
\begin{proof}
    \begin{align*}
        \Lie_{X_H} f &= \Lie_{X_H}(g - \contr{Y}\eta) = \contr{X_H}\d g - \contr{Y}\Lie_{X_H}\eta - \contr{[X_H,Y]}\eta\\
        &= \contr{X_H}\d g + \contr{Y}( \Rz(H)\eta + \Rt(H)\tau ) + \contr{[Y,X_H]}\eta\\
        &= \contr{X_H}\d g + \Rz(H)\contr{Y}\eta + \Rt(H)\contr{Y}\tau + \Lie_Y \contr{X_H}\eta - \contr{X_H}\Lie_Y\eta\\
        &= \contr{X_H}\d g + \Rz(H)\contr{Y}\eta + \Rt(H)\contr{Y}\tau - \Lie_Y H - \contr{X_H}(\rho\eta + \d g)\\
        &= \Rz(H)\contr{Y}\eta - \rho H - g\Rz(H) - \rho \contr{Y}\eta = -(g - \contr{X_H}\eta)\Rz(H)\\
        &= -\Rz(H) f\,.
    \end{align*}
\end{proof}



\begin{prop}
    If $Y$ is a $(\rho, g)$-Cartan symmetry, then $Z= Y - g \Rz$ is a generalized infinitesimal dynamical symmetry.
\end{prop}
\begin{proof}   
    Suppose that $Y$ is a $(\rho, g)$-Cartan symmetry. Then, by Theorem \ref{thm_Cartan_dissipated}, the function $f= g - \contr{Y}\eta$ is a dissipated quantity, so, by Theorem \ref{Noether_thm}, $Z=X_f-\Rt$ is a generalized infinitesimal dynamical symmetry. The Hamiltonian vector field of $f$ is given by
    \begin{align}
        \flat (X_f) = \d g - \d (\contr{Y} \eta) 
        - \left( \liedv{\Rz} g - \liedv{\Rz} \contr{Y} \eta + g - \contr{Y} \eta \right) \eta
        + \left(1 - \liedv{\Rt} g + \liedv{\Rt} \contr{Y} \eta \right) \tau\,,
    \end{align}
    but
    \begin{align}
         \liedv{\Rz} \contr{Y} \eta 
       & = \contr{[\Rz, Y]} \eta + \contr{Y} \liedv{Rz}  \eta   
        = \contr{[\Rz, Y]} \eta
        = - \contr{[Y, \Rz]} \eta
        = -\liedv{Y} \contr{\Rz} \eta + \contr{\Rz} \liedv{Y} \eta
        =  \contr{\Rz} \liedv{Y} \eta\\
        &= \contr{\Rz} \left(\rho \eta + \d g\right)
        = \rho + \liedv{\Rz} g\, ,
    \end{align}
    and, similarly, $\liedv{\Rt} \contr{Y} \eta = \liedv{\Rt} g$. 
    In addition, 
    \begin{align}
        \d (\contr{Y} \eta)  = \liedv{Y} \eta - \contr{Y} \d \eta
        = \rho \eta + \d g - \contr{Y} \d \eta\, .
    \end{align}
    Thus,
    \begin{align}
        \flat (X_f) = \contr{Y} \d \eta
        - \left(g - \contr{Y} \eta \right) \eta
        + \tau\,.
    \end{align}
    On the other hand,
    \begin{equation}
        \flat (Y) = (\contr{Y} \eta) \eta + \contr{Y} \d \eta\,,
    \end{equation}
    so we can write
    \begin{equation}
        \flat (X_f - Y) = -g\eta + \tau\,,
    \end{equation}
    that is,
    \begin{equation}
        X_f = Y - g \Rz + \Rt\,,
    \end{equation}
    so $Z= Y - g \Rz$.
\end{proof}

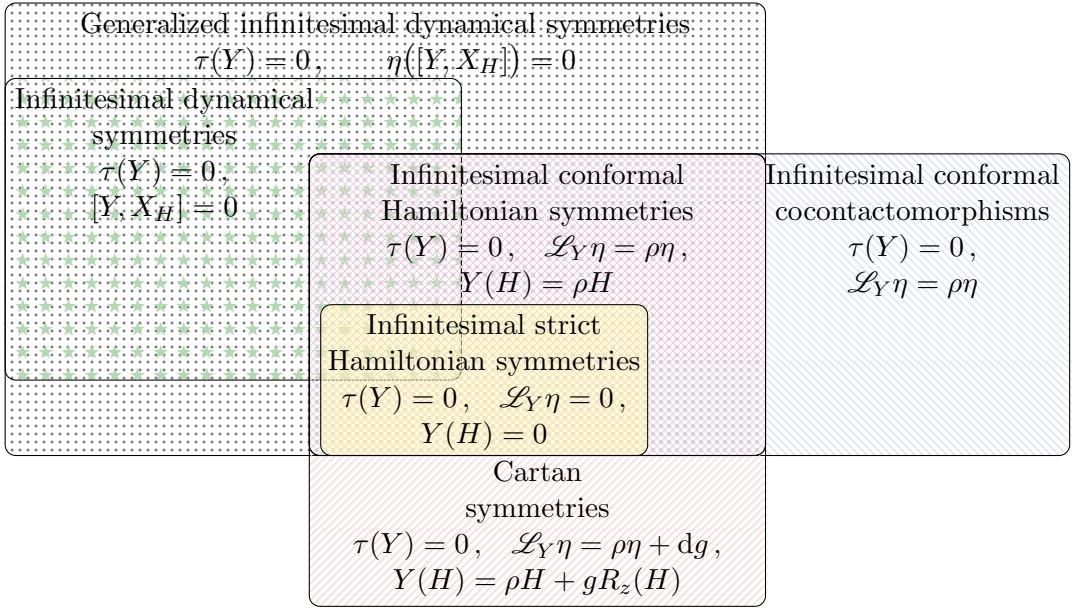
\begin{figure}[t]
    \centering  
  \begin{tikzpicture}[every text node part/.style={align=center}]
      \draw (0,0) node[minimum height=6cm,minimum width=10cm,draw, pattern=dots, pattern color=gray, rounded corners] (gen) {};
      \node[below] at (gen.north) {Generalized infinitesimal dynamical symmetries\\ $\tau(Y)=0\,,\qquad \eta\big([Y, X_H]\big)=0$};
      \draw (-2,0) node[minimum height=4cm, minimum width = 6cm, fill, draw, pattern=fivepointed stars, pattern color=green!30, rounded corners] (dyn) {};
      \node[below right] at (dyn.north west) {Infinitesimal dynamical\\ symmetries\\ $\tau(Y)=0\,,$\\ $[Y, X_H]=0$};
      \draw (4,-1) node[minimum height=4cm, minimum width = 10cm, draw, pattern=north west lines, pattern color=blue!15, rounded corners](coco){};
      \node[below left] at (coco.north east) {Infinitesimal conformal\\cocontactomorphisms\\$\tau(Y)=0\,,$\\$\liedv{Y}\eta=\rho \eta$};
      \draw (2,-2) node[minimum height=6cm, minimum width = 6cm, draw, pattern=north east lines, pattern color=red!20, rounded corners](cartan){};
      \node[above] at (cartan.south) {Cartan\\ symmetries\\$\tau(Y)=0\,,$\quad $\liedv{Y}\eta=\rho \eta+\d g\,,$\\ $Y(H)=\rho H +g \Rz(H)$};
      \draw (2,-1) node[minimum height=4cm, minimum width = 6cm, draw, rounded corners](conham){};
      \node[below] at (conham.north) {Infinitesimal conformal\\ Hamiltonian symmetries\\$\tau(Y)=0\,,$\quad $\liedv{Y}\eta=\rho \eta\,,$\\ $Y(H)=\rho H$};
      \draw (1.3,-2) node[color=yellow!40, minimum height=2cm, minimum width = 4.3cm, fill, semitransparent, rounded corners](strham){};
      \draw (1.3,-2) node[minimum height=2cm, minimum width = 4.3cm, draw, rounded corners](strham){};
      \node[below] at (strham.north) {Infinitesimal strict\\ Hamiltonian symmetries\\$\tau(Y)=0\,,$\quad$\liedv{Y}\eta=0\,,$\\ $Y(H)=0$};
    \end{tikzpicture}
  
    \caption{Classification of infinitesimal symmetries and relations between them. Infinitesimal dynamical symmetries, infinitesimal conformal Hamiltonian symmetries and infinitesimal strict Hamiltonian symmetries close Lie algebras, whereas Cartan symmetries and generalized infinitesimal dynamical symmetries do not close Lie algebras.}
    \label{fig:infinitesimal_symmetries}
  
\end{figure}

\begin{rmrk}
    If $Y$ is a $(\rho, g)$-Cartan symmetry and $Z= Y-g \Rz$ is its associated generalized infinitesimal dynamical symmetry, then the dissipated quantities associated to $Y$ and to $Z$ via Theorems \ref{Noether_thm} and \ref{thm_Cartan_dissipated} coincide.
\end{rmrk}

Regarding the Lie algebra structures formed by the sets of symmetries, we have the following result:

\begin{prop}[Lie algebras of symmetries]\phantom{m}
    \begin{enumerate}[{\rm(1)}]
        \item Infinitesimal conformal Hamiltonian symmetries close a Lie subalgebra of $(\mathfrak{X}(M), [\cdot,\cdot])$. More precisely, if $Y_1$ is a $\rho_1$-conformal Hamiltonian symmetry and $Y_2$ is a $\rho_2$-conformal Hamiltonian symmetry, then $[Y, Z]$ is a $\widetilde \rho$-conformal Hamiltonian symmetry, where $\widetilde \rho = Y_1(\rho_2) - Y_2(\rho_1)$.
        \item Infinitesimal strict Hamiltonian symmetries close a Lie subalgebra from the Lie algebra of infinitesimal conformal Hamiltonian symmetries.
    \end{enumerate}
\end{prop}

\begin{proof}
        If $Y_i$ is a $(\rho_i, g_i)$-Cartan symmetry (for $i=1,2$), then
        \begin{align}
            \liedv{[Y_1, Y_2]} \eta 
            & = \liedv{Y_1} \liedv{Y_2} \eta -  \liedv{Y_2} \liedv{Y_1} \eta 
            = \liedv{Y_1} \left( \rho_2 \eta + \d g_2  \right)
            - \liedv{Y_2} \left( \rho_1 \eta + \d g_1  \right)\\
            & = \left(Y_1 (\rho_2) - Y_2(\rho_1)  \right) \eta 
            + \d \left(Y_1(g_2) -Y_2(g_1) \right) 
            + \rho_2 \d g_1 + \rho_1 \dd g_2,
        \end{align}
        so, in general, $[Y_1, Y_2]$ is not a Cartan symmetry (see Example \ref{ex:lieCartan}). However, for $g_1=g_2= 0$, 
        \begin{equation}
            \liedv{[Y_1, Y_2]} \eta = \left(Y_1 (\rho_2) - Y_2(\rho_1)  \right) \eta = \widetilde {\rho} \eta\,.
        \end{equation}
        Moreover,
         \begin{equation}
           \liedv{[Y_1, Y_2]} H = \liedv{Y_1} \liedv{Y_2} H -  \liedv{Y_2} \liedv{Y_1} H
            = \liedv{Y_1} \left( \rho_2 H \right)
            - \liedv{Y_2} \left( \rho_1 H  \right) = \left(Y_1 (\rho_2) - Y_2(\rho_1)  \right) H\,,
        \end{equation}
        and hence $[Y_1, Y_2]$ is an infinitesimal $\widetilde \rho$-conformal Hamiltonian symmetry.
        
        In particular, if $Y_1$ an $Y_2$ are infinitesimal strict Hamiltonian symmetries, then $\rho_1=\rho_2\equiv 0$, so $\widetilde \rho\equiv 0$ and thus $[Y_1, Y_2]$ is an infinitesimal strict Hamiltonian symmetry.
\end{proof}

In general, Cartan symmetries do not close a Lie subalgebra.
\begin{exmpl}\label{ex:lieCartan}
    Consider the cocontact Hamiltonian system $(\R^4, \tau, \eta, H)$, with $\tau= \dd t,\ \eta = \dd z - p \dd q$ and 
    $$H=e^{q-z} ,$$
    where $(t, q, p, z)$ are the canonical coordinates in $\R^4$. 
    The vector field 
    $$
    Y_1=q\frac{\partial}{\partial z}
    $$
    is a $(0,q)$-Cartan symmetry and
    $$
    Y_2=(p-1)e^{q-z}\frac{\partial}{\partial p}-e^{q-z}\frac{\partial}{\partial z}
    $$
    is a $(e^{q-z},0)$-Cartan symmetry. Their commutator is $[Y_1,Y_2]=-qY_2$, and
    $$
    \Lie_{[Y_1,Y_2]}\eta=-qe^{q-z}\eta+e^{q-z}\dd q\,.
    $$
    There is no function $f\in \Cinfty(\R^4)$ such that $f\eta+e^{q-z}\dd q$ is exact,
    so it is not possible to write $\Lie_{[Y_1,Y_2]}\eta=\rho \eta + \dd g$ for any functions $\rho, g\in \Cinfty(\R^4)$, and hence $[Y_1, Y_2]$ is not a Cartan symmetry.
\end{exmpl}



The types of symmetries and the relations between them are summarized in Figure \ref{fig:infinitesimal_symmetries}.

\section{Symmetries and dissipated quantities of cocontact Lagrangian systems}
\label{sec_symmetries_Lagrangian}



Consider a regular cocontact Lagrangian system $(\R \times \T Q \times \R, \L)$, with cocontact structure $(\d t, \eta_\L)$. Since $(\R \times \T Q \times \R, \d t, \eta_\L, E_\L)$ is a cocontact Hamiltonian system, every result from Section \ref{section_symmetries_Hamiltonian} can be applied to this case. Moreover, making use of the geometric structures of the tangent bundle \cite{Yano1973,deLeon1989} (and their natural extensions to $\R \times \T Q \times \R$) we can consider additional types of symmetries. A summary of these symmetries and their relations can be found in Figure~\ref{fig:infinitesimal_symmetries_Lagrangian}. The relation between (extended) natural symmetries of the Lagrangian and Hamiltonian symmetries is depicted in Figure~\ref{fig:infinitesimal_symmetries_Lagrangian_Hamiltonian}.


Consider a diffeomorphism $\varphi=(\varphi_Q, \varphi_z)\colon Q \times \R \to Q\times \R$, where $\varphi_Q\colon Q \to Q$ and $\varphi_z\colon \R \to \R$ are diffeomorphisms (in an abuse of notation we omit the projections).
Then, the {\emph{action-dependent lift}} of $\varphi$ is the diffeomorphism $\widetilde \varphi = (\id_\R, \T \varphi_Q, \varphi_z): \R \times \T Q \times \R \to \R \times \T Q \times \R$. A vector field $Y \in \mathfrak{X}(Q\times \mathbb{R})$ is \emph{split} if it is projectable by $\text{pr}_Q:Q\times\mathbb{R}\rightarrow Q$ and by $\text{pr}_{\mathbb{R}}:Q\times\mathbb{R}\rightarrow\mathbb{R}$. Given a split vector field $Y \in \mathfrak{X}(Q\times \mathbb{R})$, its {\emph{action-dependent lift}} is the vector field $\bar Y^C \in \mathfrak{X}(\R \times \T Q \times R)$ whose local flow is the action-dependent lift of the local flow of $Y$.  In other words, if $Y$ is locally of the form

\begin{equation}\label{eq:local-Y}
Y=Y^i(q) \frac{\partial}{\partial q^i}+ \zeta(z) \frac{\partial}{\partial z}\,,
\end{equation}
its {\emph{action-dependent complete lift}} is the vector field given locally by
$$\bar Y^C=Y^i(q) \frac{\partial}{\partial q^i}
+ v^j \frac{\partial Y^i}{\partial q^j} \frac{\partial}{\partial v^i}
+ \zeta(z) \frac{\partial}{\partial z}\, .$$


Given a function $f\in \Cinfty(Q)$, its \emph{vertical lift} is the function
$f^V=f\circ \tau_Q \circ \tau_2\in \Cinfty(\R \times \T Q \times \R)$, where $\tau_Q\circ\tau_2:\R \times \T Q \times \R\rightarrow Q$ is the projection (see Section~\ref{sec:lag}). A 1-form $\omega \in \Omega^1(Q)$ can be regarded as a function $\widehat \omega \in \Cinfty(\T Q)$. Locally, if $\omega= \omega_i (q) \d q^i$, then $\widehat \omega= \omega_i(q) v^i$. 
The {vertical lift} of a vector field $X\in \mathfrak{X}(Q)$ to $\T Q$ is the unique vector field $X^V \in \mathfrak{X}(\T Q)$ such that $X^V(\widehat \omega)=(\omega(X))^V$ for any $\omega \in \Omega^1(Q)$. The \emph{vertical lift} of an split $Y\in \mathfrak{X}(Q\times \R)$ to $\R \times \T Q \times \R$ is the vector field $\bar Y^V\in \mathfrak{X}(\R \times \T Q \times \R)$ given by the vertical lift of $\T \pr{Q} Y\in \mathfrak{X}(Q)$ to $TQ$. Locally, if $Y$ has the local expression \eqref{eq:local-Y}, its vertical lift reads
$$ \bar Y^V= Y^i(q) \parder{}{v^i}\,. $$ 
The following properties hold for any $X, Y \in \mathfrak{X}(Q\times \R)$:
\begin{equation} \label{eq:properties_lifts}
    [\bar X^C, \Delta] = 0\,,\qquad \mathcal{S} (\bar X^C) = \bar X^V\,,\qquad\mathcal{S} (\bar X^V) = 0\,,\qquad
    \Lie_{\bar X^V}\mathcal{S} = 0\,,\qquad\Lie_{\bar X^C}\mathcal{S} = 0\,,
\end{equation}
where $\mathcal{S}$ and $\Delta$ denote the vertical endomorphism and the Liouville vector field, with local expressions~\eqref{eq:structures_tangent}.



\subsection{Lagrangian symmetries}
 We will denote $\phi'\equiv\dfrac{\d \phi}{\d z}$. Henceforth, all the Lagrangian systems are assumed to be regular.

\begin{dfn}
    A diffeomorphism
    $\Phi\colon \R \times \T Q \times \R\to \R \times \T Q \times \R$ of the form
    $$\Phi\colon (t, q, v, z) \mapsto \left(t, \Phi_q(t,q,v), \Phi_v(t,q,v), \Phi_z(z) \right)$$ is called an \emph{extended symmetry of the Lagrangian} if $\Phi^\ast \L = \Phi_z'\L$. In addition, if $\Phi$ is the action-dependent lift of some $\varphi \in \Diff(Q\times \R)$, then it is called an \emph{extended natural symmetry of the Lagrangian}.
    
    A vector field $Y \in \mathfrak{X}(\R \times \T Q \times \R)$ of the form
    $$ Y=A^i(t,q, v) \parder{}{q^i} + B^i(t,q,v) \parder{}{v^i}+ \zeta(z) \parder{}{z}$$
    is called an \emph{infinitesimal extended symmetry of the Lagrangian} if $\liedv{Y} \L = \zeta' \L$. 
    In addition, if $Y$ is the action-dependent complete lift of some $X \in \mathfrak{X}(Q\times \R)$, then it is called an \emph{infinitesimal extended natural symmetry of the Lagrangian}.
\end{dfn}

\begin{prop}\label{prop:extLagSym}
    An (infinitesimal) extended natural symmetry $\bar Y^C$ of the Lagrangian $\L$ is an (infinitesimal) $\Phi_z'$-conformal ($\zeta'$-conformal) Hamiltonian symmetry of the cocontact Hamiltonian system $(M,\tau,\eta_\L,E_\L)$.
\end{prop}
\begin{proof}
Clearly, $\contr{\bar Y^C}\tau=0$. Moreover,
$$
\Lie_{\bar Y^C}E_\L=\Lie_{\bar Y^C}(\Delta(\L))- \Lie_{\bar Y^C}(L)=(\Delta-1)(\Lie_{\bar Y^C}(\L))=(\Delta-1)(\zeta'\L)=\zeta'E_{\L}\,,
$$
where we have used that the action-dependent complete lift of a vector field commutes with the Liouville vector field (see properties \eqref{eq:properties_lifts}), and
\begin{align*}
 \Lie_{\bar Y^C}\eta_\L&=\Lie_{\bar Y^C}(\d z- \transp{\cal S} \circ \d \L)=\d \zeta-\transp{\cal S} \circ \d \left(\Lie_{\bar Y^C}\L\right)=\zeta'\d z-\transp{\cal S} \circ \left(\L\d \zeta' +\zeta'\d\L\right)
 \\
 &=\zeta'\left(\d z- \transp{\cal S} \circ \d \L\right)=\zeta'\eta_\L\,.   
\end{align*}

Therefore, $\bar Y^C$ is a $\zeta'$-conformal Hamiltonian symmetry. The case for extended natural symmetries of the Lagrangian is similar.
\end{proof}

\begin{prop}\label{prop:extsym}
    Let $Y=Y^i(q) \tparder{}{q^i}+ \zeta(z) \tparder{}{z}$ be an split vector field on $Q\times \mathbb{R}$. Then $\bar Y^C$ is an infinitesimal extended natural symmetry of $\L$ if, and only if, $\bar Y^V(\L)-\zeta$ is a dissipated quantity.
\end{prop}

\begin{proof}
    We have that
    $$\eta_\L(\bar Y^C)=(\d z- \transp{\cal S} \circ \d \L) (\bar Y^C)=\zeta-\bar Y^V(\L)\, ,$$
    where we have used the second of the properties~\eqref{eq:properties_lifts},
    so
    \begin{align*}
    \liedv{\Gamma_\L} \left(\bar Y^V(\L)-\zeta\right)&+\Lie_{\Rz^\L}(E_\L) \left(\bar Y^V(\L)-\zeta\right) 
    = -\liedv{\Gamma_\L} \contr{\bar Y^C} \eta_\L - \contr{\bar Y^C} \big( \Rz^\L (E_\L) \eta_\L \big)  \\
    & = -\liedv{\Gamma_\L} \contr{\bar Y^C} \eta_\L + \contr{\bar Y^C} \big( \liedv{\Gamma_\L} \eta_\L +\Rt(E_\L) \tau \big)  
    =-\contr{[\Gamma_L, \bar Y^C]} \eta_\L\, .
    \end{align*}
    If $\Gamma_L$ is the Herglotz--Euler--Lagrange vector field (given by equations~\eqref{eq:Euler-Lagrange-field}), 
    \begin{align*}
        \contr{[\bar Y^C, \Gamma_L ]} \eta_\L
        &= \liedv{\bar Y^C} \contr{\Gamma_\L} \eta_\L - \contr{\Gamma_\L} \liedv{\bar Y^C} \eta_\L= -\liedv{\bar Y^C} E_\L - \contr{\Gamma_\L} \liedv{\bar Y^C} \left( \d z - \transp{\cal S} \circ \d \L \right)\\
        &= -\Delta (\liedv{\bar Y^C} \L)+ \liedv{\bar Y^C} \L- \contr{\Gamma_\L} \liedv{\bar Y^C}  \d z+ \contr{\Gamma_\L} \transp{\cal S} \circ \d \left(\liedv{\bar Y^C} \L\right) \\
        &= \liedv{\bar Y^C} \L - \contr{\Gamma_\L} \liedv{\bar Y^C}  \d z
        = \liedv{\bar Y^C} \L - \contr{\Gamma_\L} \d \zeta
        = \liedv{\bar Y^C} \L - \liedv{\Gamma_\L} \zeta\, ,
        \end{align*}
    where $\contr{\Gamma_\L}\transp{\cal S}=\Delta$ because $\Gamma_\L$ is a {\sc sode}. Thus, $\bar Y^V(\L)-\zeta$ is a dissipated quantity if and only if $\liedv{\bar Y^C} \L - \liedv{\Gamma_\L} \zeta$ vanishes.
\end{proof}

A particular case of extended natural symmetries are those with $\zeta=0$. That is, symmetries which are lifted from $Q$.

\begin{dfn}
    A diffeomorphism $\Phi \in \Diff (\R \times \T Q \times \R)$ is called a \emph{symmetry of the Lagrangian} if $\Phi^\ast \L = \L$ and $\Phi^\ast t=t$. In addition, if $\Phi$ is the canonical lift of some $\varphi \in \Diff(Q)$, then it is called a \emph{natural symmetry of the Lagrangian}.
    
    A vector field $Y \in \mathfrak{X}(\R \times \T Q \times \R)$ is called an \emph{infinitesimal symmetry of the Lagrangian} if $\liedv{Y} \L = 0$ and $\contr{Y} \tau=0$. In addition, if $Y$ is the complete lift of some $X \in \mathfrak{X}(Q)$, then it is called an \emph{infinitesimal natural symmetry of the Lagrangian}.
\end{dfn}


From Proposition \ref{prop:extLagSym}, we have the following.
\begin{cor}
    Every (infinitesimal) natural symmetry of the Lagrangian $\L$ is an (infinitesimal) strict Hamiltonian symmetry of $(\R \times \T Q\times \R, \d t, \eta_{\L}, E_\L)$.
\end{cor}

It is worth noting that a symmetry of the Lagrangian which is not natural is not, in general, a Hamiltonian symmetry. Moreover, in general, it is not an extended symmetry of the Lagrangian either.

\begin{exmpl}
   Consider the Lagrangian $L(t,x,v,z)=\frac{1}{2} v^2-V(t,x,z)$ on $\R\times \T\R\times \R$. Clearly, the vector field
   $$Y = v \frac{\partial }{\partial x} + \frac{\partial V}{\partial x} \frac{\partial }{\partial v}$$
   is an infinitesimal symmetry of the Lagrangian (but it is not natural). However, $Y(E_\L) \neq 0$. Moreover, we have $\eta_\L = \d z - v \d x$, so
   $$\liedv{Y} \eta_\L 
= - \frac{\partial V}{\partial x} \d x - v \d v \neq \rho \eta_\L\, ,  
$$ 
for any $\rho\in \mathcal{C}^\infty(M)$
\end{exmpl}
From Proposition \ref{prop:extsym} we have that:
\begin{cor}\label{cor_natural_Lagrangian}
    Let $Y$ be a vector field on $Q$ and assume that $\L$ is regular. Then $Y^C$ is an infinitesimal natural symmetry of $\L$ if, and only if, $Y^V(\L)$ is a dissipated quantity.
\end{cor}

\begin{exmpl}[Cyclic coordinate]
   Suppose that $\L$ has a cyclic coordinate, namely $\tparder{L}{q^i}=0$ for some $i \in \left\{1, \ldots, n\right\}$. Then, $\bar{Y}^C$ is an infinitesimal natural Lagrangian symmetry, where $Y= \tparder{}{q^i}$, and its associated dissipated quantity is the corresponding momentum $\tparder{\L}{v^i}$.
\end{exmpl}

\begin{prop}
    Infinitesimal symmetries of the Lagrangian, infinitesimal natural symmetries of the Lagrangian and infinitesimal extended natural symmetries of the Lagrangian close Lie subalgebras of $(\mathfrak{X}(\R \times \T Q \times \R), [\cdot, \cdot])$.
\end{prop}

\begin{proof}
    If $Y_1, Y_2 \in \mathfrak{X}(\R \times \T Q \times \R)$ are symmetries of the Lagrangian $\L$, then
    \begin{align}
        & \liedv{[Y_1, Y_2]} \L = \left[\liedv{Y_1}, \liedv{Y_2} \right] \L = 0,\\
        & \contr{[Y_1, Y_2]} \tau = 0,
    \end{align}
    so $[Y_1, Y_2]$ is a symmetry of the Lagrangian. In particular, if $Y_1=X_1^C$ and $Y_2=X_2^C$ (for some $X_1, X_2\in \mathfrak{X}(Q)$) are natural symmetries of the Lagrangian, then $[Y_1, Y_2] = [X_1, X_2]^C$. Therefore, $[Y_1, Y_2]$ is also a natural symmetry of the Lagrangian.
    
    Similarly, suppose that $\bar Y_1^C$ and $\bar Y_2^C$ are extended natural symmetries of the Lagrangian $\L$, where
    $$Y_a = Y_a^i(q) \frac{\partial}{\partial q^i} + \zeta_a(z) \frac{\partial}{\partial z},\qquad a =1,2\, . $$
    Then,
    $$\liedv{[\bar Y_1^C, \bar Y_2^C]} \L 
    = \left[ \liedv{\bar Y_1^C}, \liedv{\bar Y_2^C}\right] \L
    = \left( \zeta_1 \zeta_2^{\prime \prime} - \zeta_2 \zeta_1^{\prime \prime}\right) \L
    = \frac{\d}{\d z}\left( \zeta_1 \zeta_2^{\prime } - \zeta_2 \zeta_1^{\prime }\right) \L
    \, ,$$
    but
    $$ [Y_1, Y_2] 
    = \left(Y_1^i \frac{\partial Y_2^j}{\partial q^i} - Y_2^i \frac{\partial Y_1^j}{\partial q^i}\right) \frac{\partial}{\partial q^j}
    + \left( \zeta_1 \zeta_2^{\prime} - \zeta_2 \zeta_1^{\prime }\right) \frac{\partial}{\partial z},
    $$
    so $[\bar Y_1^C, \bar Y_2^C]$ is an extended natural symmetry of $\L$. 
    
\end{proof}

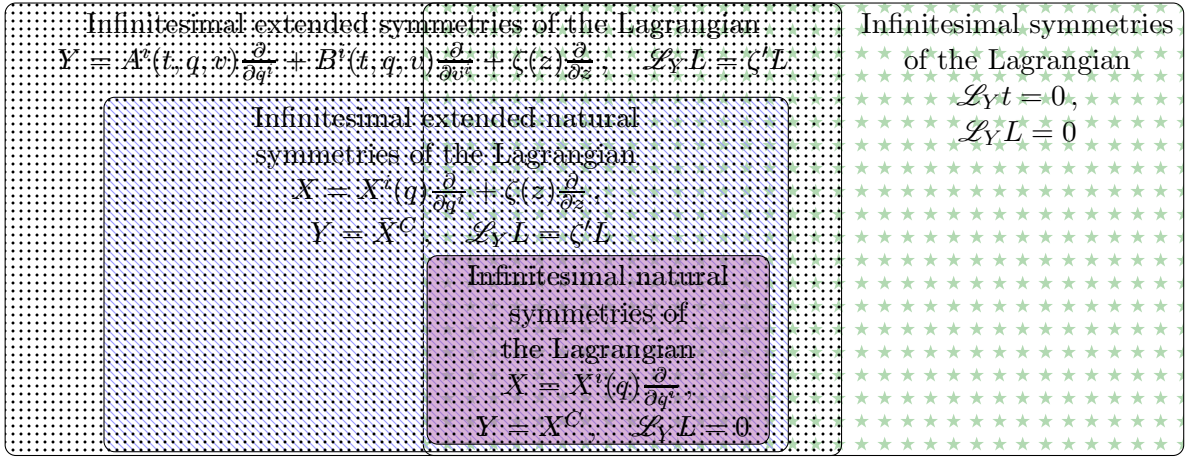
\begin{figure}[t]
    \centering  
  \begin{tikzpicture}[every text node part/.style={align=center}]
    \draw (5,0) node[minimum height=6cm, minimum width = 10cm, fill, draw, pattern=fivepointed stars, pattern color=green!30, rounded corners](coco){};
    \node[below left] at (coco.north east) {Infinitesimal symmetries\\of the Lagrangian\\$\liedv{Y} t=0\,,$\\$\liedv{Y}L=0$};
    \draw (0.3,-0.6) node[minimum height=4.7cm, minimum width = 9cm,, draw, pattern=north west lines, pattern color=blue!40, rounded corners](conham){};
    \node[below] at (conham.north) {Infinitesimal extended natural\\ symmetries of the Lagrangian\\$X=X^i(q) \frac{\partial}{\partial q^i}+ \zeta(z) \frac{\partial}{\partial z}$\,,\\ $\quad Y=\bar X^C\,,\quad \liedv{Y}L =\zeta' L$};
    \draw (2.3,-1.6) node[color=violet!60, minimum height=2.5cm, minimum width = 4.5cm, fill, semitransparent, rounded corners](strham){};
    \draw (2.3,-1.6) node[minimum height=2.5cm, minimum width = 4.5cm, draw, rounded corners](strham){};
    \node[below] at (strham.north) {Infinitesimal natural\\ symmetries of \\the Lagrangian\\$X=X^i(q) \frac{\partial}{\partial q^i}\,,$\\ $\quad Y=X^C\,,\quad \liedv{Y}L =0$};
    \draw (0,0) node[minimum height=6cm,minimum width=11cm,draw, pattern=dots, rounded corners] (gen) {};
    \node[below] at (gen.north) {Infinitesimal extended symmetries of the Lagrangian\\ $Y=A^i(t,q, v) \parder{}{q^i} + B^i(t,q,v) \parder{}{v^i}+ \zeta(z) \parder{}{z}\,,\quad \liedv{Y} L = \zeta' L$};
    \end{tikzpicture}
  
    \caption{Classification of infinitesimal Lagrangian symmetries and relations between them. Infinitesimal symmetries of the Lagrangian, infinitesimal natural symmetries of the Lagrangian and infinitesimal extended natural symmetries of the Lagrangian close Lie subalgebras.}
    \label{fig:infinitesimal_symmetries_Lagrangian}
\end{figure}

\begin{figure}[t]
    \centering  
  \begin{tikzpicture}[every text node part/.style={align=center}]
    \draw (0.3,-0.6) node[minimum height=4.7cm, minimum width = 9cm,, draw, pattern=north west lines, pattern color=blue!40, rounded corners](conham){};
    \node[below] at (conham.north) {Infinitesimal extended natural\\ symmetries of the Lagrangian\\$X=X^i(q) \frac{\partial}{\partial q^i}+ \zeta(z) \frac{\partial}{\partial z}$\,,\\ $\quad Y=\bar X^C\,,\quad \liedv{Y}L =\zeta' L$};
    \draw (1.3,-2.2) node[minimum height=8cm, minimum width = 12cm, draw, pattern=north east lines, pattern color=red!60, rounded corners](cartan){};
    \node[above] at (cartan.south) {Infinitesimal $\rho$-conformal Hamiltonian symmetries\\$\liedv{Y} t=0\,,\quad \liedv{Y}\eta_L=\rho \eta_L\,,$\\ $\liedv{Y}E_L=\rho E_L$};
    \draw (3,-2.3) node[color=yellow!60, minimum height=4cm, minimum width = 8cm, fill, semitransparent, rounded corners](strham){};
    \draw (3,-2.3) node[minimum height=4cm, minimum width = 8cm, draw, rounded corners](strham){};
    \node[above] at (strham.south) {Infinitesimal strict Hamiltonian symmetries\\$\liedv{Y} t=0\,,\quad\liedv{Y}\eta_L=0\,,\quad \liedv{Y}E_L=0$};
    \draw (2.3,-1.6) node[color=violet!60, minimum height=2.5cm, minimum width = 4.5cm, fill, semitransparent, rounded corners](strham){};
    \draw (2.3,-1.6) node[minimum height=2.5cm, minimum width = 4.5cm, draw, rounded corners](strham){};
    \node[below] at (strham.north) {Infinitesimal natural\\ symmetries of \\the Lagrangian\\$X=X^i(q) \frac{\partial}{\partial q^i}\,,$\\ $\quad Y=X^C\,,\quad \liedv{Y}L =0$};
    \end{tikzpicture}
  
    \caption{Relations between infinitesimal (extended) natural symmetries of the Lagrangian, conformal Hamiltonian symmetries and strict Hamiltonian symmetries.}
    \label{fig:infinitesimal_symmetries_Lagrangian_Hamiltonian}
\end{figure}
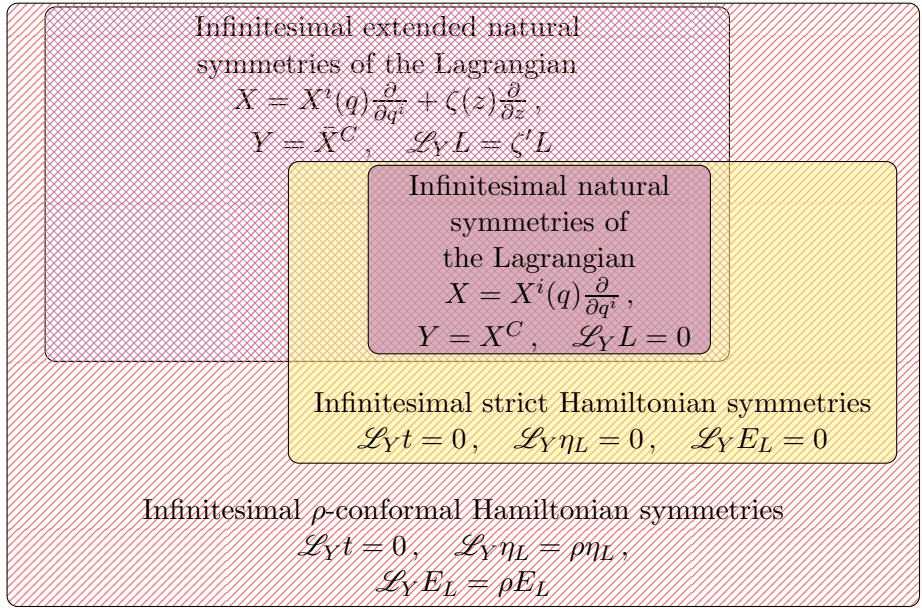

\subsection{Symmetries of the action}

Another relevant class of symmetry are transformations on the ``$z$" variable, or changes of action, which preserve the dynamics. This kind of transformations are used in \cite{de_leon_inverse_2022} to generate equivalent Lagrangians.
\begin{dfn}
    A diffeomorphism $\Phi\colon \R \times \T Q \times \R\to \R \times \T Q \times \R$  is a \emph{change of action} if, for any section $\gamma$ of the projection $\pr{\R \times \T Q}\colon \R \times \T Q \times \R \to \R \times \T Q$, we have
    $$\pr{\R \times \T Q}\circ\,\Phi\circ \gamma=\mathrm{Id}_{\R \times \T Q}\, .$$
    A vector field $Z \in \mathfrak{X}(\R\times \T Q \times \R)$ an \emph{infinitesimal change of action} if $\T\pr{\R \times \T Q} \circ Z=0$.
\end{dfn}

If a change of action has the form
    $$\Phi\colon (t, q, v, z) \mapsto \left(t, q, v, \Phi_z(t,q,v,z) \right)\,,$$
then, in particular $\dfrac{\partial \Phi_z}{\partial z}\neq0$ everywhere.

Clearly, the flow of an infinitesimal change of action is made up of changes of action. Moreover, if $Y\in\mathfrak{X}(\mathbb{R}\times Q\times\mathbb{R})$ is a {\sc sode} and $\Phi$ is a change of action, then $\Phi_*Y$ is also a {\sc sode}. 

    \begin{prop}
    A change of action
    $\Phi\colon \R \times \T Q \times \R\to \R \times \T Q \times \R$ of the form
    $$\Phi\colon (t, q, v, z) \mapsto \left(t, q, v, \Phi_z(t,q,v,z) \right)\,,$$
    is a generalized dynamical symmetry if, and only if,  
    $\Gamma_L(\Phi_z)=\L\circ\Phi$.
    
    An infinitesimal change of action $Z \in \mathfrak{X}(\R\times \T Q \times \R)$ with local expression
    $$
    Z = \zeta(t,q,v, z) \parder{}{z}
    $$
    is a generalized infinitesimal dynamical symmetry if, and only if, $\zeta$ is a dissipated quantity, i.e., $\Gamma_\L (\zeta) = \zeta \tparder{L}{z}$.    
    \end{prop}
    \begin{proof}
Given two {\sc sode} $Y$ and $X$, we have that $\theta_L(Y)=\theta_L(X)=\Delta(L)$ and $^t\mathcal{S}(Y)=\Delta$. Let $\Gamma_L$ be the Herglotz--Euler--Lagrange vector field of the system, given by equations~\eqref{eq:Euler-Lagrange-field}.
If $\Phi$ is a change of action, then
$$ \Phi_*\Gamma_L=\Gamma_L^t\parder{}{t}+\Gamma_L^{q^i}\parder{}{q^i}+\big(\Gamma_L^{v^i}\circ\Phi^{-1}\big)\parder{}{v^i}+\big(\Gamma_L(\Phi_z)\circ\Phi^{-1}\big)\frac{\partial}{\partial z}\,.$$ 
In addition,
$$\contr{\Phi_*\Gamma_L}\eta_\L=\Gamma_L(\Phi_z)\circ\Phi^{-1}-\theta_\L(\Gamma_\L)=\Gamma_L(\Phi_z)\circ\Phi^{-1}-\Delta(\L)
\,.
$$
On the other hand, $\contr{\Gamma_L}\eta_\L=-E_{\L}=\L-\Delta(\L)$. Therefore, $\Phi$ is a generalized dynamical symmetry (i.e. $\contr{\Phi_*\Gamma_L}\eta=\contr{\Gamma_L}\eta$) if, and only if, $\Gamma_L(\Phi_z)= \L\circ\Phi $.


Furthermore, if $Z$ is an infinitesimal change of action we have that
\begin{align}
    \contr{[\Gamma_\L, Z]} \eta_\L 
    &= \liedv{\Gamma_\L} \contr{Z} \eta_\L - \contr{Z} \liedv{\Gamma_\L} \eta_\L 
    = \Gamma_\L(\zeta) + \contr{Z} \big( \Rz^\L (E_\L) \eta_\L +\Rt^\L(E_\L) \tau\big)\\
    &= \Gamma_\L(\zeta) + \Rz^\L (E_\L) \zeta.
\end{align}
And the result is proved using the identity $\tparder{L}{z}=-\Rz^\L (E_\L)$
\end{proof}
This result motivates the following definition.

\begin{dfn}
    A diffeomorphism
    $\Phi\colon \R \times \T Q \times \R\to \R \times \T Q \times \R$ of the form
    $$\Phi\colon (t, q, v, z) \mapsto \left(t, q, v, \Phi_z(t,q,v,z) \right)$$
    is an \emph{action symmetry} if $\Gamma_L(\Phi_z)=\L\circ\Phi $.

    A vector field $Z \in \mathfrak{X}(\T Q \times \R)$ of the form $Z=\zeta(t, q, v, z) \tparder{}{z}$ is an \emph{infinitesimal action symmetry} if $\zeta$ is a dissipated quantity.
\end{dfn}

\section{Examples}
\label{sec_examples}

We compute several examples to illustrate in practice some of the concepts presented previously. We also show how symmetries and dissipated quantities can be used to study the dynamics of the $2$-body problem with time-dependent friction.

\subsection{The free particle with time-dependent mass and linear dissipation}
    Consider the cocontact Hamiltonian system $(\R\times \cT \R \times \R, \d t, \eta, H)$, {with natural coordinates $(t, q, p, z)$ where
    $\eta=\d z -p \d q$ is the contact form and
    \begin{equation}
        H = \frac{p^2}{2m(t)} + \frac{\kappa}{m(t)} z\, ,
    \end{equation}
    is the Hamiltonian function,} with $m$ a positive-valued function depending only on $t$, expressing the mass of the particle, and $\kappa$ a positive constant. The Hamiltonian vector field of $H$ is
    $$ X_H =\parder{}{t} + \frac{p}{m(t)}\parder{}{q} -  p\frac{\kappa}{m(t)} \parder{}{p} + \left(\frac{p^2}{2m(t)} - \frac{\kappa}{m(t)} z\right)\parder{}{z}\,. $$
    {
    Its integral curves are given by
    \begin{equation}
       \begin{dcases}
    		\dot q = \frac{p}{m(t)}\,,
    		\\
    	   \dot p = -  p\frac{\kappa}{m(t)} \,,
    		\\
    	   \dot z =  \frac{p^2}{2m(t)} - \frac{\kappa}{m(t)} z\,,		
    	\end{dcases}	  
    \end{equation}
    which yield
    \begin{equation}
        \begin{dcases}
         & q(t) =  \int _0^t\frac{\exp \left(\int _0^u-\frac{\kappa }{m(s)}\dd s\right) p_0}{m(u)}\dd u+q_0\, ,\\
         & p(t) = p_0 \exp \left(\int _0^t-\frac{\kappa }{m(s)}\dd s\right)\, ,\\
         & z(t) = \exp \left(\int _0^t-\frac{\kappa }{m(v)}\dd v\right) \int _0^t\frac{\exp \left( \int _0^w-\frac{\kappa }{m(s)}\dd s\right) p_0^2}{2 m(w)}\dd w+z_0 \exp \left(\int _0^t-\frac{\kappa }{m(v)}\dd v\right)\, ,
        \end{dcases}
    \end{equation}
    where $q_0=q(0),\, p_0=p(0),\, z_0=z(0)$ are the initial conditions. The term of $H$ linear in the variable $z$ permits to model a damping phenomena. As a matter of fact, in the particular case where $m(t)$ is constant the linear momenta (and hence the velocity) of the system decreases exponentially.
    }

    The function $f(t,q,p,z)=\exp \left(- \displaystyle \int _0^t\frac{\kappa}{m(s)}\d s\right)$ is a dissipated quantity. Hence, by Theorem \ref{Noether_thm}, the vector field
    \begin{equation}
        Y_f = X_f - \Rt 
        = -\exp \left(- \int _0^t\frac{\kappa}{m(s)} \d s\right) \parder{}{z}
        = -\exp \left(- \int _0^t\frac{\kappa}{m(s)} \d s\right)  \Rz
    \end{equation}
    is a generalized infinitesimal dynamical symmetry. In addition, one can verify that $Y_f$ is an infinitesimal dynamical symmetry, namely $Y_f$ commutes with $X_H$. Now,
    \begin{equation}
        Y_f (H) 
        = -\exp \left(- \int _0^t\frac{\kappa}{m(s)} \d s\right) \Rz(H)\, ,
    \end{equation}
    and
    \begin{equation}
        \liedv{Y_f} \eta
        = \dd \left( -\exp \left(- \int _0^t\frac{\kappa}{m(s)} \d s\right) \contr{\Rz} \eta \right)
        = -\dd \left(\exp \left(- \int _0^t\frac{\kappa}{m(s)} \d s\right) \right)\, ,
    \end{equation}
    so $Y_f$ is a $(0,g)$-Cartan symmetry, where $g=-\exp \left(-\displaystyle\int _0^t\frac{\kappa}{m(s)} \d s\right)$.

    Moreover,
    $f_2(t,q,p,z)=p$ is also a dissipated quantity, whose associated generalized infinitesimal dynamical symmetry is
    $$ Y_{f_2} = \parder{}{q}\,. $$
    It is clear that $Y_{f_2}$ is an infinitesimal dynamical symmetry, i.e., $Y_{f_2}$ commutes with $X_H$. Moreover, $\liedv{Y_{f_2}} \eta = 0$
    and $Y_{f_2}(H)=0$, so $Y_{f_2}$ is an infinitesimal strict Hamiltonian symmetry.

    The Lagrangian counterpart of this system is characterized by the Lagrangian function $L\colon \R \times \T \R \times \R \to \R$ given by 
    \begin{equation}
        \L = m(t)\frac{v^2}{2} - \frac{\kappa}{m(t)} z\,.
    \end{equation}
    The vector field $Z\in \mathfrak{X}(\R\times \T \R \times \R)$ with local expression
    \begin{equation}
        Z = \zeta(t,q, v,z) \parder{}{z}
        = \exp \left(- \int _0^t\frac{\kappa}{m(s)} \d s\right) \parder{}{z}
    \end{equation}
    is an infinitesimal action symmetry, since it is an infinitesimal change of action and we know that $\zeta$ is a dissipated quantity.

\subsection{An action-dependent central potential with time-dependent mass}
Consider a Lagrangian function $L\colon \R \times \T \R^2 \times \R \to\R$ of the form
\begin{equation}
    \L {(t, x, y, v_x, v_y, z)} = \frac{m(t)}{2} \left(v_x^2 + v_y^2\right) - V \left(t,(x^2+y^2), z\right)\, ,
\end{equation}
where $m(t)$ is a positive-valued function. Let $Y\in \mathfrak{X}(\R^2)$ be infinitesimal generator of rotations on the plane, namely,
\begin{equation}
    Y = -y \parder{}{x} + x \parder{}{y}\, .
\end{equation}
Its complete lift is given by
\begin{equation}
    \bar{Y}^C = -y \parder{}{x} + x \parder{}{y} - v_y \parder{}{v_x} + v_x \parder{}{v_y} \, ,
\end{equation}
and its vertical lift is 
\begin{equation}
    \bar{Y}^V = -y \parder{}{v_x} + x \parder{}{v_y}\, .
\end{equation}
Clearly, $\bar{Y}^C$ is an infinitesimal natural symmetry of the Lagrangian, i.e., $\bar{Y}^C(\L)=0$. Hence, by Corollary \ref{cor_natural_Lagrangian},
$$\bar{Y}^V(\L) = m(t)\left(-y v_x + x v_y\right)$$
is a dissipated quantity. This quantity is the angular momentum for a particle with time-dependent mass.

\subsection{The two-body problem with time-dependent friction}

The $2$-body problem describes the dynamics of two particles under the effects of a force that depends on the distance between the particles, usually the gravitational force. To model time-dependent friction, we will add a linear term on the action in the Lagrangian, with a time-dependent coefficient. {The two-body problem is one of the most important problems in celestial mechanics. The addition of a friction term may allow to describe the motion of celestial bodies in a dissipative medium.}

The phase space is $\mathbb{R}\times \T\mathbb{R}^6\times \mathbb{R}$, with coordinates $(t,\bm{q}^1,\bm{q}^2,\bm{v}^1,\bm{v}^2,z)$. The superindex denotes each particle, and the bold notation is a shorthand for the three spatial components, namely $\bm{q}^1=(q_1^1,q_2^1,q_3^1)$ and $\bm{q}^2=(q_1^2,q_2^2,q_3^2)$. The relative distance between the particles is $\bm{r}=\bm{q}^2-\bm{q}^1$, whose (Euclidean) length will be denoted $r = \vert\bm{r}\vert$.

The Lagrangian function is
$$
L=\frac12m_1\bm{v}^1\cdot\bm{v}^1+\frac12m_2\bm{v}^2\cdot\bm{v}^2-U(r)-\gamma(t)z\,,
$$
where $m_1,m_2\in\mathbb{R}$ are the masses of the particles which we assume to be constant, $U(r)$ is the central potential and $\gamma$ is a time-dependent function. The Lagrangian energy is
$$
E_L=\frac12m_1\bm{v}^1\cdot\bm{v}^1+\frac12m_2\bm{v}^2\cdot\bm{v}^2+U(r)+\gamma(t)z\,,
$$
and the cocontact structure is given by the one-forms
$$
\eta=\d z-m_1\bm{v}^1\cdot\d\bm{q}^1-m_2\bm{v}^2\cdot\d\bm{q}^2\,, \quad \tau=\d t\,.
$$
The evolution of the system is given by the Herglotz--Euler--Lagrange vector field $\Gamma_L$, defined by equations~\eqref{eq:Euler-Lagrange-field} and with local expression~\eqref{eq:Euler-Lagrange-field_local}.
Its solutions satisfy the Herglotz--Euler--Lagrange equations:
\begin{align}
       m_1\dot{\bm{v}}^1&=\bm{F}-\gamma(t) m_1\bm{v}^1\,,\label{eq:HEL2body1}
    \\
    m_2\dot{\bm{v}}^2 &= -\bm{F}-\gamma(t) m_2\bm{v}^2\,. \label{eq:HEL2body2}
\end{align}
The dot notation indicates time derivative and $\displaystyle\bm{F}=-\frac{\d U}{\d r}\frac{\bm{r}}{r}$ is the force of the potential $U$. 

Proceeding as in the classical $2$-body problem, we study the evolution of the center of masses
$$
\bm{R}=\frac{m_1\bm{q}^1+m_2\bm{q}^2}{m_1+m_2}\,.
$$
Since $\Gamma_L$ is a {\sc sode}, we have that
$$
\Gamma_L(\bm{R})=\frac{m_1\bm{v}^1+m_2\bm{v}^2}{m_1+m_2}=\dot{\bm{R}}\,,
$$
and
$$
\Gamma_L(\dot{\bm{R}})=-\gamma\dot{\bm{R}}\,.
$$
That is, every component of $\dot{\bm{R}}$ is a dissipated quantity. Along a solution, it evolves as
$$
\dot{\bm{R}}(t)=\dot{\bm{R}}_0e^{-\int \gamma(t)\d t}\,.
$$
In particular, if $\gamma$ is a positive constant, as the time increases the center of mass tends to move on a line with constant speed $\dot{\bm{R}}_0$. By Noether's Theorem \ref{Noether_thm}, the corresponding generalized infinitesimal dynamical symmetries are $\bm{Y_{\dot{R}}}=X_{\dot{\bm{R}}}-R^L_t$, where $R^L_t$ is subtracted to every component. A short computation shows that
$$
\bm{Y_{\dot{R}}}=\frac{1}{m_1+m_2}\left(\frac{\partial}{\partial \bm{q}^1}+\frac{\partial}{\partial \bm{q}^2}\right)\,.
$$
Each component of $\bm{Y_{\dot{R}}}$ is an action dependent complete lift and $\Lie_{\bm{Y_{\dot{R}}}}L=0$ therefore, they are infinitesimal natural symmetries of the Lagrangian.

The fact that the center of mass is moving in a very concrete way, may indicate that one could express the system using only the relative position. Indeed, from equations~\eqref{eq:HEL2body1} and \eqref{eq:HEL2body2} one derives
$$
\mu\ddot{\bm{r}}=-\bm{F}-\gamma\mu\dot{\bm{r}}\,,
$$
where $\mu=\dfrac{m_1m_2}{m_1+m_2}$ is the reduced mass.  This equation can also be derived from the Lagrangian $L_\mu=\frac{1}{2}\mu\dot{\bm{r}}\cdot\dot{\bm{r}}-U(r)-\gamma z$. The angular momentum is 
$$
\bm{L}=\mu\bm{r}\times\dot{\bm{r}}\,.
$$
Each component is a dissipated quantity:
$$
\Gamma_{L}(\bm{L})=-\gamma\bm{L}\,.
$$
The angular momentum along a solution is
$$
\bm{L}(t)=\bm{L}_0e^{-\int \gamma(t)\d t}\,.
$$
Since the direction of $\bm{L}$ remains constant, the movement takes place on a plane perpendicular to $\bm{L}_0$. If $\gamma$ is a positive constant, the angular momentum tends to $0$. The associated generalized infinitesimal dynamical symmetries are
$$
\bm{Y_{L}}=X_{\bm{L}}-R^L_t=\bm{r}\times\left(\frac{1}{m_2}\frac{\partial}{\partial \bm{q}^2}-\frac{1}{m_1}\frac{\partial}{\partial \bm{q}^1}\right)-\dot{\bm{r}}\times\left(\frac{1}{m_2}\frac{\partial}{\partial \bm{v}^2}-\frac{1}{m_1}\frac{\partial}{\partial \bm{v}^1}\right)
$$
Each component of $\bm{Y_{L}}$ is an action dependent complete lift and $\Lie_{\bm{Y_{L}}}L=0$, therefore they are infinitesimal natural symmetries of the Lagrangian.

Finally, the Lagrangian energy $E_L$ evolves as
$$
\Gamma_L(E_L)=-R_z^L(E_L)E_L+R_t^L(E_L)=-\gamma E_L+\dot{\gamma}z\,,
$$
and it is not a dissipated quantity due to the time-dependence of $\gamma$.

The evolution of the mechanical energy, namely the sum of the kinetic and the potential energies,
$$E_{\rm mec}=\frac12m_1\bm{v}^1\cdot\bm{v}^1+\frac12m_2\bm{v}^2\cdot\bm{v}^2+U(r)\,$$ 
is given by
$$
\Gamma_L(E_{\rm mec})=-\gamma(t)\left(m_1\bm{v}^1\cdot\bm{v}^1+m_2\bm{v}^2\cdot\bm{v}^2\right)\,.
$$
We could proceed by rewriting the reduced system in polar coordinates and describe the possible orbits. Unfortunately, in this case it is not evident how to express the relation between the radial and angular coordinates.
\section{Conclusions and further research}
\label{sec_conclsuions}

In this paper, we have characterized the symmetries and dissipated quantities of time-dependent contact Hamiltonian and Lagrangian systems. Firstly, we have studied generalized infinitesimal dynamical symmetries, a type of symmetries which are in bijection with dissipated quantities. After that, we have considered other types of symmetries which preserve (up to a conformal factor) additional objects, such as the cocontact structure or the Hamiltonian function. Moreover, making use of the canonical structures of the tangent bundle, we have discussed Lagrangian symmetries and symmetries of the action. We have concluded with three illustrative examples: the free particle with time-dependent mass and linear dissipation, the action-dependent central potential with time-dependent mass, and the two-body problem with time-dependent friction.

{In particular, the two-body problem could be interesting in celestial mechanics, where the friction could be used to model the damping caused by the medium. The formalism presented in this paper may also be applied to more complex systems in celestial mechanics. In a future work, we plan to extend this study to the restricted three-body problem with friction. It would be particularly interesting to study how the friction affects the stability of the system.}

The study of symmetries and dissipated quantities made in this work is the first step towards investigating the symmetries and dissipation laws in non-conservative field theories using the $k$-(co)contact \cite{Gaset2020,Gaset2021a,Riv-2022} and multicontact \cite{LGMRR-2022} settings. Furthermore, the classification of symmetries could provide a new insight towards a reduction method for time-(in)dependent contact systems.

\section*{Acknowledgements}
We wish to express our gratitude to the referee for his/her valuable comments. J.~Gaset and X.~Rivas acknowledge financial support of the Ministerio de Ciencia, Innovaci\'on y Universidades (Spain), projects PGC2018-098265-B-C33 and D2021-125515NB-21. 
A.~López-Gordón received financial support from the Spanish Ministry of Science and Innovation (MCIN/AEI/ 10.13039/501100011033), under the grants PID2019-106715GB-C21 and CEX2019-000904-S and the predoctoral contract PRE2020-093814. 
X.~Rivas also acknowledges financial support of  the Novee Idee 2B-POB II project PSP: 501-D111-20-2004310 funded by the ``Inicjatywa Doskonałości - Uczelnia Badawcza'' (IDUB) program.


\let\emph\oldemph

\printbibliography

@article{Riv-2022,
    author = {Xavier Rivas},
    title = {{Nonautonomous $k$-contact field theories}},
    year = {2023},
    journal = {J. Math. Phys.},
    volume = {64},
    number = {3},
    pages = {033507},
    doi = {10.1063/5.0131110}
}

@article{Bravetti2017,
  title = {Contact {{Hamiltonian Dynamics}}: {{The Concept}} and {{Its Use}}},
  shorttitle = {Contact {{Hamiltonian Dynamics}}},
  author = {Bravetti, Alessandro},
  date = {2017-10},
  journaltitle = {Entropy},
  volume = {19},
  number = {10},
  pages = {535},
  publisher = {{Multidisciplinary Digital Publishing Institute}},
  issn = {1099-4300},
  doi = {10.3390/e19100535},
  url = {https://www.mdpi.com/1099-4300/19/10/535},
  urldate = {2022-07-12},
  issue = {10},
  langid = {english}
}

@article{Bravetti2017a,
  title = {Contact {{Hamiltonian}} Mechanics},
  author = {Bravetti, Alessandro and Cruz, Hans and Tapias, Diego},
  date = {2017-01-01},
  journaltitle = {Ann. Phys.},
  volume = {376},
  pages = {17--39},
  issn = {0003-4916},
  doi = {10.1016/j.aop.2016.11.003},
  url = {https://www.sciencedirect.com/science/article/pii/S0003491616302469},
  urldate = {2022-07-12},
  langid = {english}
}

@article{Bravetti2019,
  title = {Contact Geometry and Thermodynamics},
  author = {Bravetti, Alessandro},
  date = {2019-02},
  journaltitle = {Int. J. Geom. Methods Mod. Phys.},
  volume = {16},
  pages = {1940003},
  publisher = {{World Scientific Publishing Co.}},
  issn = {0219-8878},
  doi = {10.1142/S0219887819400036},
  url = {https://www.worldscientific.com/doi/abs/10.1142/S0219887819400036},
  urldate = {2022-07-12},
  issue = {supp01}
}

@article{Carinena1989,
  title = {A New Approach to the Converse of {{Noether}}'s Theorem},
  author = {Carinena, J. F. and Lopez, C. and Martinez, E.},
  date = {1989-11},
  journaltitle = {J. Phys. A: Math. Gen.},
  volume = {22},
  number = {22},
  pages = {4777--4786},
  publisher = {{IOP Publishing}},
  issn = {0305-4470},
  doi = {10.1088/0305-4470/22/22/009},
  url = {https://doi.org/10.1088/0305-4470/22/22/009},
  urldate = {2021-01-02},
  langid = {english}
}

@article{Carinena1989a,
  title = {Symmetry Theory and {{Lagrangian}} Inverse Problem for Time-Dependent Second-Order Differential Equations},
  author = {Carinena, J. F. and Martinez, E.},
  date = {1989-07},
  journaltitle = {J. Phys. A: Math. Gen.},
  volume = {22},
  number = {14},
  pages = {2659--2665},
  publisher = {{IOP Publishing}},
  issn = {0305-4470},
  doi = {10.1088/0305-4470/22/14/016},
  url = {https://doi.org/10.1088/0305-4470/22/14/016},
  urldate = {2021-01-02},
  langid = {english}
}

@article{Carinena1994,
  title = {A Geometrical Version of {{Noether}}'s Theorem in Supermechanics},
  author = {Cariñena, José F. and Figueroa, Héctor},
  date = {1994},
  journaltitle = {Rep. Math. Phys.},
  volume = {34},
  number = {3},
  pages = {277--303},
  issn = {0034-4877},
  doi = {10.1016/0034-4877(94)90002-7},
  url = {https://mathscinet.ams.org/mathscinet-getitem?mr=1339466},
  urldate = {2021-01-15},
  mrnumber = {1339466}
}

@article{Ciaglia2018,
  title = {Contact Manifolds and Dissipation, Classical and Quantum},
  author = {Ciaglia, F. M. and Cruz, H. and Marmo, G.},
  date = {2018-11-01},
  journaltitle = {Ann. Phys.},
  volume = {398},
  pages = {159--179},
  issn = {0003-4916},
  doi = {10.1016/j.aop.2018.09.012},
  url = {https://www.sciencedirect.com/science/article/pii/S0003491618302574},
  urldate = {2022-07-12},
  langid = {english}
}

@book{deLeon1989,
  title = {Methods of Differential Geometry in Analytical Mechanics},
  author = {de León, Manuel and Rodrigues, Paulo R.},
  options = {useprefix=true},
  date = {1989},
  series = {North-{{Holland}} Mathematics Studies},
  number = {158},
  publisher = {{North-Holland}},
  location = {{Amsterdam ; New York : New York, N.Y., U.S.A}},
  isbn = {978-0-444-88017-8},
  langid = {english},
  pagetotal = {483}
}

@article{deLeon1994,
  title = {Classification of Symmetries for Higher Order {{Lagrangian}} Systems.},
  author = {de León, Manuel and Martín de Diego, David},
  options = {useprefix=true},
  date = {1994},
  journal = {Extracta Mathematicae},
  volume = {9},
  number = {1},
  pages = {32-36},
  publisher = {{Universidad de Extremadura}},
  issn = {0123-8743},
  url = {https://digital.csic.es/handle/10261/2246},
  urldate = {2020-10-23},
  langid = {english}
}

@article{deLeon1996,
  title = {Symmetries and Constants of the Motion for Singular {{Lagrangian}} Systems},
  author = {de Léon, Manuel and Martín de Diego, David},
  options = {useprefix=true},
  date = {1996},
  journaltitle = {Internat. J. Theoret. Phys.},
  volume = {35},
  number = {5},
  pages = {975--1011},
  issn = {0020-7748},
  doi = {10.1007/BF02302383},
  url = {https://mathscinet.ams.org/mathscinet-getitem?mr=1386775},
  urldate = {2020-09-22},
  mrnumber = {1386775}
}

@article{deLeon2017,
  title = {Cosymplectic and Contact Structures for Time-Dependent and Dissipative {{Hamiltonian}} Systems},
  author = {de León, Manuel and Sardón, Cristina},
  options = {useprefix=true},
  date = {2017-06},
  journaltitle = {J. Phys. A: Math. Theor.},
  volume = {50},
  number = {25},
  pages = {255205},
  publisher = {{IOP Publishing}},
  issn = {1751-8121},
  doi = {10.1088/1751-8121/aa711d},
  url = {https://doi.org/10.1088/1751-8121/aa711d},
  urldate = {2021-05-01},
  langid = {english}
}

@article{deLeon2019,
  title = {Contact {{Hamiltonian}} Systems},
  author = {de León, Manuel and Lainz Valcázar, Manuel},
  options = {useprefix=true},
  date = {2019-10},
  journaltitle = {J. Math. Phys.},
  volume = {60},
  number = {10},
  pages = {102902},
  publisher = {{American Institute of Physics}},
  issn = {0022-2488},
  doi = {10.1063/1.5096475},
  url = {https://aip.scitation.org/doi/10.1063/1.5096475},
  urldate = {2022-07-12}
}

@article{deLeon2019a,
  title = {Singular {{Lagrangians}} and Precontact {{Hamiltonian}} Systems},
  author = {de León, Manuel and Lainz Valcázar, Manuel},
  options = {useprefix=true},
  date = {2019-10},
  journaltitle = {Int. J. Geom. Methods Mod. Phys.},
  volume = {16},
  number = {10},
  pages = {1950158},
  publisher = {{World Scientific Publishing Co.}},
  issn = {0219-8878},
  doi = {10.1142/S0219887819501585},
  url = {https://www.worldscientific.com/doi/abs/10.1142/S0219887819501585},
  urldate = {2022-07-12}
}

@misc{deLeon2020a,
%   title = {Optimal Control, Contact Dynamics and {{Herglotz}} Variational Problem},
%   author = {de León, Manuel and Lainz, Manuel and Muñoz-Lecanda, Miguel C.},
%   options = {useprefix=true},
%   date = {2020-06-25},
%   number = {arXiv:2006.14326},
%   eprint = {2006.14326},
%   eprinttype = {arxiv},
%   primaryclass = {math-ph},
%   publisher = {{arXiv}},
%   doi = {10.48550/arXiv.2006.14326},
%   url = {http://arxiv.org/abs/2006.14326},
%   urldate = {2022-07-13},
%   archiveprefix = {arXiv},
%   keywords = {49J15; 37J55 (Primary) 80M50; 70Q05 (Secondary),Mathematical Physics,Mathematics - Optimization and Control}
% }

@article{deLeon2020b,
  title = {Infinitesimal Symmetries in Contact {{Hamiltonian}} Systems},
  author = {de León, Manuel and Lainz Valcázar, Manuel},
  options = {useprefix=true},
  date = {2020-07},
  journaltitle = {J. Geom. Phys.},
  volume = {153},
  pages = {103651},
  issn = {03930440},
  doi = {10.1016/j.geomphys.2020.103651},
  url = {https://linkinghub.elsevier.com/retrieve/pii/S0393044020300486},
  urldate = {2021-11-17},
  langid = {english}
}

@article{deLeon2021a,
  title = {Contact {{Hamiltonian}} and {{Lagrangian}} Systems with Nonholonomic Constraints},
  author = {de León, Manuel and Jiménez, Víctor M. and Lainz, Manuel},
  options = {useprefix=true},
  date = {2021},
  journaltitle = {J. Geom. Mech.},
  volume = {13},
  number = {1},
  pages = {25},
  publisher = {{American Institute of Mathematical Sciences}},
  doi = {10.3934/jgm.2021001},
  url = {https://www.aimsciences.org/article/doi/10.3934/jgm.2021001},
  urldate = {2022-07-12},
  langid = {english}
}

@misc{deLeon2021b,
  title = {A Review on Contact {{Hamiltonian}} and {{Lagrangian}} Systems},
  author = {de León, Manuel and Lainz, Manuel},
  options = {useprefix=true},
  date = {2021-02-22},
  eprint = {2011.05579},
  eprinttype = {arxiv},
  primaryclass = {math-ph},
  url = {http://arxiv.org/abs/2011.05579},
  urldate = {2021-03-25},
  archiveprefix = {arXiv}
}

@article{deLeon2021c,
  title = {Symmetries, Constants of the Motion, and Reduction of Mechanical Systems with External Forces},
  author = {de León, Manuel and Lainz, Manuel and López-Gordón, Asier},
  options = {useprefix=true},
  date = {2021-04-01},
  journaltitle = {J. Math. Phys.},
  volume = {62},
  number = {4},
  pages = {042901},
  publisher = {{American Institute of Physics}},
  issn = {0022-2488},
  doi = {10.1063/5.0045073},
  url = {https://aip.scitation.org/doi/10.1063/5.0045073},
  urldate = {2021-04-03}
}

@article{deLeon2021e,
  title = {The {{Hamilton}}–{{Jacobi Theory}} for {{Contact Hamiltonian Systems}}},
  author = {family=León, given=Manuel, prefix=de, useprefix=true and Lainz, Manuel and Muñiz-Brea, Álvaro},
  date = {2021},
  journaltitle = {Mathematics},
  volume = {9},
  number = {16},
  pages = {1993},
  publisher = {{Multidisciplinary Digital Publishing Institute}},
  doi = {10.3390/math9161993},
  url = {https://www.mdpi.com/2227-7390/9/16/1993},
  urldate = {2022-01-05},
  issue = {16},
  langid = {english}
}

@article{Az1,
title = {{Canonical and canonoid transformations for Hamiltonian systems on (co)symplectic and (co)contact manifolds}},
author = {R. Azuaje and  A.M. Escobar-Ruiz},
journal = {J. Math. Phys.},
volume = {64},
number = {3},
doi = {10.1063/5.0135045},
year = {2023},
pages = {033501}
}

@misc{Az2,
      title={Lie integrability for time-independent and time-dependent Hamiltonian systems}, 
      author={R. Azuaje},
      year={2023},
      eprint={2302.02218},
      archivePrefix={arXiv},
      primaryClass={math-ph}
}

@misc{deLeon2022,
%   title = {Time-Dependent Contact Mechanics},
%   author = {de León, Manuel and Gaset, Jordi and Gràcia, Xavier and Muñoz-Lecanda, Miguel Carlos and Rivas, Xavier},
%   options = {useprefix=true},
%   date = {2022-05-19},
%   number = {arXiv:2205.09454},
%   eprint = {2205.09454},
%   eprinttype = {arxiv},
%   primaryclass = {math-ph},
%   publisher = {{arXiv}},
%   url = {http://arxiv.org/abs/2205.09454},
%   urldate = {2022-07-12},
%   archiveprefix = {arXiv},
%   keywords = {37J55; 70H03; 70H05; 53D05; 53D10; 53Z05; 70H45,Mathematical Physics,Mathematics - Symplectic Geometry}
% }

@article{deLeon2022a,
  title = {Hamilton--{{Jacobi}} theory and integrability for autonomous and non-autonomous contact systems},
  shorttitle = {Hamilton--{{Jacobi}} Theory for Contact Systems},
  author = {de León, Manuel and Lainz, Manuel and López-Gordón, Asier and Rivas, Xavier},
  options = {useprefix=true},
  journal = {J. Geom. Phys.},
  volume = {187},
  pages = {104787},
  doi = {10.1016/j.geomphys.2023.104787},
  year = {2023}
}

@article{Djukic1975,
  title = {Noether's Theory in Classical Nonconservative Mechanics},
  author = {Djukic, Dj. S. and Vujanovic, B. D.},
  date = {1975-03-01},
  journaltitle = {Acta Mechanica},
  volume = {23},
  number = {1},
  pages = {17--27},
  issn = {1619-6937},
  doi = {10.1007/BF01177666},
  url = {https://doi.org/10.1007/BF01177666},
  urldate = {2020-09-29},
  langid = {english}
}

@article{Eberard2007,
  title = {An Extension of {{Hamiltonian}} Systems to the Thermodynamic Phase Space: {{Towards}} a Geometry of Nonreversible Processes},
  shorttitle = {An Extension of {{Hamiltonian}} Systems to the Thermodynamic Phase Space},
  author = {Eberard, D. and Maschke, B. M. and van der Schaft, A. J.},
  options = {useprefix=true},
  date = {2007-10-01},
  journaltitle = {Reports on Mathematical Physics},
  volume = {60},
  number = {2},
  pages = {175--198},
  issn = {0034-4877},
  doi = {10.1016/S0034-4877(07)00024-9},
  url = {https://www.sciencedirect.com/science/article/pii/S0034487707000249},
  urldate = {2021-06-14},
  langid = {english}
}

@article{Ferrario1990,
  title = {Symmetries and Constants of Motion for Constrained {{Lagrangian}} Systems: A Presymplectic Version of the {{Noether}} Theorem},
  shorttitle = {Symmetries and Constants of Motion for Constrained {{Lagrangian}} Systems},
  author = {Ferrario, Carlo and Passerini, Arianna},
  date = {1990},
  journaltitle = {J. Phys. A},
  volume = {23},
  number = {21},
  pages = {5061--5081},
  issn = {0305-4470},
  url = {https://mathscinet.ams.org/mathscinet-getitem?mr=1083892},
  urldate = {2021-01-15},
  mrnumber = {1083892},
  doi = {10.1088/0305-4470/23/21/040}
}

@article{Gaset2020,
  title = {A Contact Geometry Framework for Field Theories with Dissipation},
  author = {Gaset, Jordi and Gràcia, Xavier and Muñoz-Lecanda, Miguel C. and Rivas, Xavier and Román-Roy, Narciso},
  date = {2020-03-01},
  journaltitle = {Ann. Phys.},
  volume = {414},
  pages = {168092},
  issn = {0003-4916},
  doi = {10.1016/j.aop.2020.168092},
  url = {https://www.sciencedirect.com/science/article/pii/S0003491620300257},
  urldate = {2022-07-12},
  langid = {english}
}

@article{Gaset2020a,
  title = {New Contributions to the {{Hamiltonian}} and {{Lagrangian}} Contact Formalisms for Dissipative Mechanical Systems and Their Symmetries},
  author = {Gaset, Jordi and Gràcia, Xavier and Muñoz-Lecanda, Miguel C. and Rivas, Xavier and Román-Roy, Narciso},
  date = {2020-05},
  journaltitle = {Int. J. Geom. Methods Mod. Phys.},
  volume = {17},
  number = {06},
  pages = {2050090},
  publisher = {{World Scientific Publishing Co.}},
  issn = {0219-8878},
  doi = {10.1142/S0219887820500905},
  url = {https://www.worldscientific.com/doi/abs/10.1142/S0219887820500905},
  urldate = {2022-07-12}
}

@inproceedings{Gaset2021,
  title = {A {{Contact Geometry Approach}} to {{Symmetries}} in {{Systems}} with {{Dissipation}}},
  booktitle = {Extended {{Abstracts GEOMVAP}} 2019},
  author = {Gaset, Jordi},
  editor = {Alberich-Carramiñana, Maria and Blanco, Guillem and Gálvez Carrillo, Immaculada and Garrote-López, Marina and Miranda, Eva},
  date = {2021},
  series = {Trends in {{Mathematics}}},
  pages = {71--75},
  publisher = {{Springer International Publishing}},
  location = {{Cham}},
  doi = {10.1007/978-3-030-84800-2_12},
  isbn = {978-3-030-84800-2},
  langid = {english}
}

@article{Gaset2021a,
  title = {A $k$-contact Lagrangian formulation for nonconservative field theories},
  author = {Gaset, Jordi and Gràcia, Xavier and Muñoz-Lecanda, Miguel C. and Rivas, Xavier and Román-Roy, Narciso},
  date = {2021-06-01},
  journaltitle = {Rep. Math. Phys.},
  volume = {87},
  number = {3},
  pages = {347--368},
  issn = {0034-4877},
  doi = {10.1016/S0034-4877(21)00041-0},
  url = {https://www.sciencedirect.com/science/article/pii/S0034487721000410},
  urldate = {2022-07-12},
  langid = {english}
}

@article{Gay-Balmaz2018,
  title = {From {{Lagrangian Mechanics}} to {{Nonequilibrium Thermodynamics}}: {{A Variational Perspective}}},
  shorttitle = {From {{Lagrangian Mechanics}} to {{Nonequilibrium Thermodynamics}}},
  author = {Gay-Balmaz, François and Yoshimura, Hiroaki},
  date = {2018-12-23},
  journaltitle = {Entropy},
  volume = {21},
  number = {1},
  pages = {8},
  issn = {1099-4300},
  doi = {10.3390/e21010008},
  url = {http://www.mdpi.com/1099-4300/21/1/8},
  urldate = {2020-09-24},
  langid = {english}
}

@book{Geiges2008,
  title = {An {{Introduction}} to {{Contact Topology}}},
  author = {Geiges, Hansjörg},
  date = {2008},
  series = {Cambridge {{Studies}} in {{Advanced Mathematics}}},
  publisher = {{Cambridge University Press}},
  location = {{Cambridge}},
  doi = {10.1017/CBO9780511611438},
  url = {https://www.cambridge.org/core/books/an-introduction-to-contact-topology/F851B2A2E7E78C6B9967A18A6641B40C},
  urldate = {2022-07-12},
  isbn = {978-0-521-86585-2}
}

@article{Gracia2000,
  title = {Fibre Derivatives: {{Some}} Applications to Singular Lagrangians},
  shorttitle = {Fibre Derivatives},
  author = {Gràcia, Xavier},
  date = {2000-02-01},
  journaltitle = {Rep. Math. Phys.},
  volume = {45},
  number = {1},
  pages = {67--84},
  issn = {0034-4877},
  doi = {10.1016/S0034-4877(00)88872-2},
  url = {https://www.sciencedirect.com/science/article/pii/S0034487700888722},
  urldate = {2022-07-12},
  langid = {english}
}

@unpublished{Herglotz1930,
  title = {Berührungstransformationen},
  author = {Herglotz, G.},
  date = {1930},
  location = {{University of Göttingen}},
  howpublished = {Lecture notes}
}

@book{Kosmann-Schwarzbach2011,
  title = {The {{Noether}} Theorems},
  author = {Kosmann-Schwarzbach, Yvette},
  date = {2011},
  series = {Sources and {{Studies}} in the {{History}} of {{Mathematics}} and {{Physical Sciences}}},
  publisher = {{Springer, New York}},
  doi = {10.1007/978-0-387-87868-3},
  url = {https://mathscinet.ams.org/mathscinet-getitem?mr=2761345},
  urldate = {2021-01-15},
  isbn = {978-0-387-87867-6},
  mrnumber = {2761345},
  pagetotal = {xiv+205}
}

@thesis{Lainz,
  title = {Contact {{Hamiltonian Systems}}},
  type = {PhD thesis},
  author = {Lainz, Manuel},
  location = {Universidad Autónoma de Madrid},
  year = {2022},
  langid = {english},
  url = {http://hdl.handle.net/10486/704774}
}

@book{Libermann1987,
  title = {Symplectic {{Geometry}} and {{Analytical Mechanics}}},
  author = {Libermann, Paulette and Marle, Charles-Michel},
  date = {1987},
  publisher = {{Springer Netherlands}},
  location = {{Dordrecht}},
  doi = {10.1007/978-94-009-3807-6},
  url = {http://link.springer.com/10.1007/978-94-009-3807-6},
  urldate = {2022-02-14},
  isbn = {978-90-277-2439-7},
  langid = {english}
}

@article{RiTo-2022,
    author = {Xavier Rivas and Daniel Torres},
    title = {{Lagrangian--Hamiltonian formalism for cocontact systems}},
    journal = {J. Geom. Mech.},
    volume = {{\bf 15}},
    number = {1},
    pages = {1--26},
    year = {2022},
    doi = {10.3934/jgm.2023001},
    url = {https://doi.org/10.3934/jgm.2023001},
}

@article{Lunev1990,
  title = {An Analogue of the {{Noether}} Theorem for Non-{{Noether}} and Nonlocal Symmetries},
  author = {Lunev, F. A.},
  date = {1990},
  journaltitle = {Teoret. Mat. Fiz.},
  volume = {84},
  number = {2},
  pages = {205--210},
  issn = {0564-6162},
  doi = {10.1007/BF01017679},
  url = {https://mathscinet.ams.org/mathscinet-getitem?mr=1077812},
  urldate = {2021-01-15},
  mrnumber = {1077812}
}

@article{Marmo1986,
  title = {Symmetries and Constants of the Motion in the {{Lagrangian}} Formalism on $\T Q$: Beyond Point Transformations},
  shorttitle = {Symmetries and Constants of the Motion in the {{Lagrangian}} Formalism {{onTQ}}},
  author = {Marmo, G. and Mukunda, N.},
  date = {1986-03},
  journaltitle = {Nuov Cim B},
  volume = {92},
  number = {1},
  pages = {1--12},
  issn = {1826-9877},
  doi = {10.1007/BF02729691},
  url = {http://link.springer.com/10.1007/BF02729691},
  urldate = {2020-09-24},
  langid = {english}
}

@article{Marwat2007,
  title = {Symmetries, Conservation Laws and Multipliers via Partial {{Lagrangians}} and {{Noether}}'s Theorem for Classically Non-Variational Problems},
  author = {Marwat, D. N. Khan and Kara, A. H. and Mahomed, F. M.},
  date = {2007},
  journaltitle = {Internat. J. Theoret. Phys.},
  volume = {46},
  number = {12},
  pages = {3022--3029},
  issn = {0020-7748},
  doi = {10.1007/s10773-007-9417-z},
  url = {https://mathscinet.ams.org/mathscinet-getitem?mr=2380993},
  urldate = {2021-01-15},
  mrnumber = {2380993}
}

@incollection{Mrugala2000,
  title = {Geometrical {{Methods}} in {{Thermodynamics}}},
  booktitle = {Thermodynamics of {{Energy Conversion}} and {{Transport}}},
  author = {Mrugała, R.},
  editor = {Sieniutycz, Stanislaw and De Vos, Alexis},
  date = {2000},
  pages = {257--285},
  publisher = {{Springer New York}},
  location = {{New York, NY}},
  doi = {10.1007/978-1-4612-1286-7_10},
  url = {https://doi.org/10.1007/978-1-4612-1286-7_10},
  isbn = {978-1-4612-1286-7}
}

@incollection{Neeman1999,
  title = {The Impact of {{Emmy Noether}}'s Theorems on {{XXIst}} Century Physics},
  booktitle = {The Heritage of {{Emmy Noether}} ({{Ramat-Gan}}, 1996)},
  author = {Ne'eman, Yuval},
  date = {1999},
  series = {Israel {{Math}}. {{Conf}}. {{Proc}}.},
  volume = {12},
  pages = {83--101},
  publisher = {{Bar-Ilan Univ., Ramat Gan}},
  url = {https://mathscinet.ams.org/mathscinet-getitem?mr=1665437},
  urldate = {2021-01-15},
  mrnumber = {1665437}
}

@article{Noether1971,
  title = {Invariant Variation Problems},
  author = {Noether, Emmy},
  date = {1971-01-01},
  journaltitle = {Transport Theory and Statistical Physics},
  volume = {1},
  number = {3},
  pages = {186--207},
  publisher = {{Taylor \& Francis}},
  issn = {0041-1450},
  doi = {10.1080/00411457108231446},
  url = {https://doi.org/10.1080/00411457108231446},
  urldate = {2021-01-16}
}

@article{Prince1983,
  title = {Toward a Classification of Dynamical Symmetries in Classical Mechanics},
  author = {Prince, Geoff},
  date = {1983-02},
  journaltitle = {Bulletin of the Australian Mathematical Society},
  volume = {27},
  number = {1},
  pages = {53--71},
  publisher = {{Cambridge University Press}},
  issn = {1755-1633, 0004-9727},
  doi = {10.1017/S0004972700011485},
  url = {https://www.cambridge.org/core/journals/bulletin-of-the-australian-mathematical-society/article/toward-a-classification-of-dynamical-symmetries-in-classical-mechanics/AE0C45CFF1D49AD45B4F3FC7A6C4CF27},
  urldate = {2021-01-02},
  langid = {english}
}

@article{Prince1985,
  title = {A Complete Classification of Dynamical Symmetries in Classical Mechanics},
  author = {Prince, Geoff},
  date = {1985-10},
  journaltitle = {Bulletin of the Australian Mathematical Society},
  volume = {32},
  number = {2},
  pages = {299--308},
  publisher = {{Cambridge University Press}},
  issn = {1755-1633, 0004-9727},
  doi = {10.1017/S0004972700009977},
  url = {https://www.cambridge.org/core/journals/bulletin-of-the-australian-mathematical-society/article/complete-classification-of-dynamical-symmetries-in-classical-mechanics/A81B8AF6047BA9B68048D831777C4107},
  urldate = {2021-01-02},
  langid = {english}
}

@thesis{Rivas2022,
  type = {PhD Thesis},
  title = {Geometrical Aspects of Contact Mechanical Systems and Field Theories},
  author = {Rivas, Xavier},
  date = {2022-04-25},
  eprint = {2204.11537},
  eprinttype = {arxiv},
  primaryclass = {math-ph},
  publisher = {{Universitat Politècnica de Catalunya}},
  location = {{Universitat Politècnica de Catalunya}},
  urldate = {2022-06-21},
  archiveprefix = {arXiv}
}

@article{Sarlet1981,
  title = {Generalizations of {{Noether}}’s {{Theorem}} in {{Classical Mechanics}}},
  author = {Sarlet, Willy and Cantrijn, Frans},
  date = {1981-10},
  journaltitle = {SIAM Rev.},
  volume = {23},
  number = {4},
  pages = {467--494},
  issn = {0036-1445, 1095-7200},
  doi = {10.1137/1023098},
  url = {http://epubs.siam.org/doi/10.1137/1023098},
  urldate = {2020-10-08},
  langid = {english}
}

@article{Sarlet1983,
  title = {Note on Equivalent {{Lagrangians}} and Symmetries},
  author = {Sarlet, W.},
  date = {1983-05},
  journaltitle = {J. Phys. A: Math. Gen.},
  volume = {16},
  number = {7},
  pages = {L229--L233},
  publisher = {{IOP Publishing}},
  issn = {0305-4470},
  doi = {10.1088/0305-4470/16/7/006},
  url = {https://doi.org/10.1088/0305-4470/16/7/006},
  urldate = {2021-01-02},
  langid = {english}
}

@article{Simoes2020,
  title = {Contact Geometry for Simple Thermodynamical Systems with Friction},
  author = {Simoes, Alexandre Anahory and de León, Manuel and Valcázar, Manuel Lainz and Martín de Diego, David Martín},
  options = {useprefix=true},
  date = {2020-09-30},
  journaltitle = {Proc. Math. Phys. Eng. Sci.},
  volume = {476},
  number = {2241},
  pages = {20200244},
  publisher = {{Royal Society}},
  doi = {10.1098/rspa.2020.0244},
  url = {https://royalsocietypublishing.org/doi/10.1098/rspa.2020.0244},
  urldate = {2022-07-12}
}

@article{vanderSchaft1983,
  title = {Symmetries, Conservation Laws, and Time Reversibility for {{Hamiltonian}} Systems with External Forces},
  author = {van der Schaft, A. J.},
  options = {useprefix=true},
  date = {1983-08},
  journaltitle = {Journal of Mathematical Physics},
  volume = {24},
  number = {8},
  pages = {2095--2101},
  issn = {0022-2488, 1089-7658},
  doi = {10.1063/1.525962},
  url = {http://aip.scitation.org/doi/10.1063/1.525962},
  urldate = {2020-09-28},
  langid = {english}
}

@book{Yano1973,
  title = {Tangent and Cotangent Bundles ; Differential Geometry.},
  author = {Yano, Kentaro and Ishihara, Shigeru},
  date = {1973},
  publisher = {{Marcel Dekker, Inc.}},
  location = {{New York}},
  langid = {english}
}

@article{Roman-Roy2020,
  title = {A Summary on Symmetries and Conserved Quantities of Autonomous {{Hamiltonian}} Systems},
  author = {Román-Roy, Narciso},
  date = {2020},
  journaltitle = {J. Geom. Mech.},
  volume = {12},
  number = {3},
  pages = {541--551},
  issn = {1941-4889},
  doi = {10.3934/jgm.2020009},
  url = {https://mathscinet.ams.org/mathscinet-getitem?mr=4160169},
  urldate = {2021-01-15},
  mrnumber = {4160169}
}

@article{Lopez1999,
  title = {Dynamical Symmetries, Non-{{Cartan}} Symmetries and Superintegrability of the n-Dimensional Harmonic Oscillator},
  author = {López, Carlos and Martínez, Eduardo and Rañada, Manuel F.},
  date = {1999-02},
  journaltitle = {J. Phys. A: Math. Gen.},
  volume = {32},
  number = {7},
  pages = {1241},
  issn = {0305-4470},
  doi = {10.1088/0305-4470/32/7/013},
  url = {https://dx.doi.org/10.1088/0305-4470/32/7/013},
  urldate = {2022-12-15},
  langid = {english}
}

@unpublished{LR-2022,
    title = {{Contact Lie systems: Theory and applications}},
    author = {Javier de Lucas and Xavier Rivas},
    year = {2022},
    note = {arXiv: \href{https://arxiv.org/abs/2207.04038}{2207.04038}}
}

@article{LGMRR-2022,
    title = {{Multicontact formalism for non-conservative field theories}},
    journal = {J. Phys. A: Math. Theor.},
    volume = {{\bf 56}},
    number = {2},
    pages = {025201},
    author = {M. de León and J. Gaset and M. C. Muñoz-Lecanda and X. Rivas and N. Román-Roy},
    year = {2022},
    doi = {10.1088/1751-8121/acb575}
}

@article{de_leon_inverse_2022,
	title = {Inverse problem and equivalent contact systems},
	volume = {176},
	issn = {03930440},
	url = {https://linkinghub.elsevier.com/retrieve/pii/S039304402200050X},
	doi = {10.1016/j.geomphys.2022.104500},
	language = {en},
	urldate = {2022-11-09},
	journal = {Journal of Geometry and Physics},
	author = {de León, Manuel and Gaset, Jordi and Lainz, Manuel},
	month = jun,
	year = {2022},
	pages = {104500},
}

@article{gaset_application_2022,
	title = {Application of {Herglotz}’s variational principle to electromagnetic systems with dissipation},
	volume = {19},
	issn = {0219-8878, 1793-6977},
	url = {https://www.worldscientific.com/doi/10.1142/S0219887822501560},
	doi = {10.1142/S0219887822501560},
	language = {en},
	number = {10},
	urldate = {2022-09-25},
	journal = {International Journal of Geometric Methods in Modern Physics},
	author = {Gaset, Jordi and Marín-Salvador, Adrià},
	year = {2022},
	pages = {2250156},
}

@article{gaset_variational_2022,
	title = {A variational derivation of the field equations of an action-dependent {Einstein}--{Hilbert} {Lagrangian}},
	author = {Gaset, Jordi and Mas, Arnau},
	options = {useprefix=true},
 year = {2023},
    journal = {J. Geom. Mech.},
    volume = {15},
    number = {1},
    pages = {357--374},
    doi = {10.3934/jgm.2023014},
    eprinttype = {arxiv},
    primaryclass = {math-ph},
}

@article{Bravetti2021,
  title = {A Geometric Approach to the Generalized {{Noether}} Theorem},
  author = {Bravetti, Alessandro and Garcia-Chung, Angel},
  date = {2021-02},
  journaltitle = {J. Phys. A: Math. Theor.},
  volume = {54},
  number = {9},
  pages = {095205},
  publisher = {{IOP Publishing}},
  issn = {1751-8121},
  doi = {10.1088/1751-8121/abde78},
  url = {https://dx.doi.org/10.1088/1751-8121/abde78},
  urldate = {2023-01-16},
  langid = {english}
}

\end{document}